%% file: final.tex
\documentclass[11pt]{article}

\newif\ifincludenumerics
\includenumericsfalse
\usepackage[a4paper,margin=1in]{geometry}
\usepackage{amsmath,amssymb,amsthm}
\usepackage{enumitem}
\usepackage[T1]{fontenc}
\usepackage{lmodern}
\usepackage{microtype}
\ifincludenumerics
\usepackage{graphicx}
\usepackage{booktabs}
\fi
\usepackage[colorlinks=true,linkcolor=blue,citecolor=blue,urlcolor=blue]{hyperref}

\newcommand{\manuscriptdate}{\today}
\hypersetup{pdftitle={Gatheral's Conjecture Revisited},
 pdfauthor={Vladimir Lucic},
 pdfsubject={Revised manuscript, \manuscriptdate},
 pdfkeywords={Heston model, local volatility, Markovian projection, convex order}}

\newtheorem{theorem}{Theorem}
\numberwithin{theorem}{section}
\newtheorem{proposition}[theorem]{Proposition}
\newtheorem{lemma}[theorem]{Lemma}
\newtheorem{corollary}[theorem]{Corollary}
\theoremstyle{remark}

\newtheorem*{remark*}{Remark}

\newcommand{\E}{\mathbb{E}}
\newcommand{\R}{\mathbb{R}}
\newcommand{\dd}{\,\mathrm{d}}
\newcommand{\Law}{\operatorname{Law}}
\newcommand{\Dop}{\mathcal{D}}

\title{Gatheral's Conjecture Revisited}
\author{Vladimir Lucic\\
\normalsize Imperial College London\\
\normalsize \href{mailto:vlucic@ic.ac.uk}{vlucic@ic.ac.uk}}
\date{\manuscriptdate}

\begin{document}

\maketitle
\begin{abstract}
We consider the Heston model with perfect negative spot--variance
correlation and its one-dimensional local volatility projection.
Let $I_T^{\mathrm H}$ and $I_T^{\mathrm{LV}}$ denote their respective
integrated variances over $[0,T]$.  We establish the inequality
\[
  \mathbb{E}\bigl[(I_T^{\mathrm H}-K)^+\bigr]
  <
  \mathbb{E}\bigl[(I_T^{\mathrm{LV}}-K)^+\bigr]
\]
for every expiry $T>0$ and every strike $K>0$.
Consequently, Heston integrated variance is strictly smaller in convex
order than the integrated variance of the calibrated local volatility
model.
This strict ordering
gives a Heston-model counterexample to the convex-order inequality
conjectured by J.\ Gatheral.
\end{abstract}

\section{Introduction}

In 2005, Gatheral~\cite{Gatheral2005} formulated the conjecture that, among
diffusive models calibrated to the same European option prices, local
volatility minimizes the value of every convex payoff of realized variance;
he later restated it in~\cite[p.~155]{Gatheral2006}.
Related practitioner calculations had already supplied implicit,
model-specific precursors in line with this claim.  Thus, for the concave volatility-swap
payoff, Brockhaus and Long~\cite{Brockhaus2000} reported higher prices in the
calibrated local volatility model than in the Heston stochastic volatility model;
see also Brockhaus et
al.~\cite[Section~12.3.2]{BrockhausEtAl2000}.
Austing~\cite[Section~10.3.4]{Austing2014} presents this
volatility-swap ordering as an empirical practitioner rule.
Numerical comparisons consistent with this ordering appear in
Le Floc'h~\cite[Section~6.3, Table~9]{LeFloch2018} for the Heston model,
and in Lipton and Reghai~\cite[Section~4.3, Table~6]{LiptonReghai2023}
for local stochastic volatility models with an exponential
Ornstein--Uhlenbeck factor.

The conjecture of Gatheral was disproved by Beiglb\"ock, Friz,
and Sturm~\cite{BeiglbockFrizSturm2011}, who provided two carefully
constructed counterexamples.
While their result settled the original claim,
it still left open the question of whether commonly used stochastic volatility models in which the
variance itself follows a Markov process can exhibit such behavior.
We show that the classical
Heston model with perfect negative correlation presents such an example.

In the remainder of this section, we informally outline the main
ideas and steps of our approach, deferring the full rigor to subsequent sections.
First, we fix some notation.
Throughout, $\R_+:=[0,\infty)$.
We use $\Rightarrow$  to denote convergence in
distribution.
For integrable random variables $X$ and $Y$, the notation
$X\preceq_{\mathrm{cx}}Y$ means
\[
 \E[\varphi(X)]\leq\E[\varphi(Y)]
\]
for every convex function $\varphi:\R\to\R$ for which both expectations
are finite, and $X\prec_{\mathrm{cx}}Y$ holds when, in addition,
$\Law(X)\ne\Law(Y)$.

Under zero interest rates and dividends, the stochastic volatility model of
Heston~\cite{Heston1993} is given by
\begin{align}
  \frac{\dd S_t}{S_t}&=\sqrt{v_t}\,\dd W_t,\label{eq:intro-heston-S}\\
  \dd v_t&=\kappa(\theta-v_t)\,\dd t
            +\xi\sqrt{v_t}\,\dd Z_t,
  \label{eq:intro-heston-v}\\
  \dd\langle W,Z\rangle_t&=\rho\,\dd t.\notag
\end{align}
Here $W$ and $Z$ are Brownian motions,
$S_0,v_0,\kappa,\theta,\xi>0$, and $\rho\in[-1,1]$.
We denote the integrated variance by
\[
  I_t^{\mathrm H}:=\int_0^t v_u\,\dd u.
\]
Its projected local variance is
\begin{equation}
  \lambda(t,s):=\E[v_t\mid S_t=s],
  \label{eq:intro-lambda}
\end{equation}
and the associated Markovian projection local volatility model and its accumulated variance are
\[
  \frac{\dd\bar S_t}{\bar S_t}
     =\sqrt{\lambda(t,\bar S_t)}\,\dd B_t,
  \qquad
  I_t^{\mathrm{LV}}
     :=\int_0^t\lambda(u,\bar S_u)\,\dd u.
\]
%
The result of Brunick and Shreve~\cite[Corollary~3.7]{BrunickShreve2013} provides at least one weak
solution $\bar S$ of this equation having the same
one-dimensional marginals as $S$.  Moreover,
$I_T^{\mathrm H}$ and $I_T^{\mathrm{LV}}$ have the same mean by the mimicking property; see \eqref{eq:equal-means}.

Henceforth, assume that $\rho=-1$, so that we may take $Z=-W$.  This is the
perfectly negatively correlated Heston model, also known as
the Heston--Nandi model; see Heston and Nandi~\cite{HestonNandi2000} and
Cox and Wang~\cite{CoxWang2013bounds}.
The main result of this note is the following convex ordering for the
Heston--Nandi model:
\begin{equation*}
  I_T^{\mathrm H}\prec_{\mathrm{cx}}I_T^{\mathrm{LV}},
  \qquad T>0.
\end{equation*}

Our proof rests on two ingredients: a backward Duhamel comparison
and a conditional form of Chebyshev's covariance
inequality.  Perfect negative correlation is what makes the two fit together.

Eliminating the common stochastic integral from
\eqref{eq:intro-heston-S}--\eqref{eq:intro-heston-v} gives the pathwise identity
\begin{equation*}
  v_t=v_0+\kappa\theta t-\xi\log(S_t/S_0)
      -\beta I_t^{\mathrm H},
  \qquad \beta:=\kappa+\frac{\xi}{2}>0.
\end{equation*}
After conditioning on the current spot and using \eqref{eq:intro-lambda},
we obtain
\begin{equation}
  \lambda(t,S_t)-v_t
   =-\beta\Bigl(\E[I_t^{\mathrm H}\mid S_t]
          -I_t^{\mathrm H}\Bigr).
  \label{eq:intro-centered}
\end{equation}
For this informal outline, we will assume that $\lambda$
is sufficiently regular to justify the calculations that follow.
For a smooth nondecreasing convex payoff $\Phi$, let
\[
 V(t,s,i):=\E_{t,s}\!\left[
 \Phi\!\left(i+\int_t^T\lambda(u,\bar S_u)\,\dd u\right)\right],
\]
where the expectation is taken for the local volatility process started at
$\bar S_t=s$.  The backward equation is
\[
 V_t+\lambda(t,s)\bigl(\tfrac12s^2V_{ss}+V_i\bigr)=0,
 \qquad V(T,s,i)=\Phi(i).
\]

Writing $H(t,s,i):=\tfrac12s^2V_{ss}(t,s,i)+V_i(t,s,i)$,
It\^o's formula gives the Duhamel identity
\begin{equation}
 \begin{aligned}
 \E[\Phi(I_T^{\mathrm{LV}})]-\E[\Phi(I_T^{\mathrm H})]
 &=V(0,S_0,0)-\E[\Phi(I_T^{\mathrm H})]\\
 &=\int_0^T\E[(\lambda(t,S_t)-v_t)
                 H(t,S_t,I_t^{\mathrm H})]\,\dd t.
 \end{aligned}
 \label{eq:intro-duhamel}
\end{equation}
%
Our approach follows the ``hedging with the wrong volatility'' argument of
El Karoui, Jeanblanc-Picqu\'e and Shreve~\cite{ElKarouiJeanblancShreve1998},
but obtains the required sign of the integrand in \eqref{eq:intro-duhamel}
through conditional averaging rather than
a pointwise ordering of the instantaneous variances.
By \eqref{eq:intro-centered}, we can write the integrand as
\[
 -\beta\E\!\left[
 \bigl(\E[I_t^{\mathrm H}\mid S_t]-I_t^{\mathrm H}\bigr)
 H(t,S_t,I_t^{\mathrm H})\right].
\]
The Feynman--Kac representation for the differentiated backward equation
shows that $H\geq0$ and that $H$ is nondecreasing in $i$
(see Lemma~\ref{lem:smooth}). On the other hand,
$\Law(S_t)$-a.s.,
\[
 f_{t,s}(i):=-\beta\bigl(\E[I_t^{\mathrm H}\mid S_t=s]-i\bigr)
\]
is increasing in $i$, since $f_{t,s}'(i)=\beta>0$, with
$\E[f_{t,s}(I_t^{\mathrm H})\mid S_t=s]=0$.
Thus,
Stroock~\cite[Exercise~1.1.9(i)]{Stroock2024} gives
\[
 \begin{aligned}
 &\E[f_{t,s}(I_t^{\mathrm H})H(t,s,I_t^{\mathrm H})\mid S_t=s]\\
 &\quad\geq\E[f_{t,s}(I_t^{\mathrm H})\mid S_t=s]\,
             \E[H(t,s,I_t^{\mathrm H})\mid S_t=s]=0,
 \end{aligned}
\]
as both factors are nondecreasing in the same variable.
Averaging over $S_t$ shows that the integrand in
\eqref{eq:intro-duhamel} is nonnegative.
Approximating call payoffs by smooth nondecreasing convex functions gives
the call inequalities;
equality of means then gives convex order
$I_T^{\mathrm H}\preceq_{\mathrm{cx}}I_T^{\mathrm{LV}}$.

Extending this argument, in our main result we prove the {\em strict} convex order at every positive expiry.
Cox and Wang~\cite{CoxWang2013bounds} establish an asymptotic extremality
property of the Heston--Nandi model relative to robust variance-option bounds
based on a prescribed terminal marginal; see also their robust Root-barrier
bounds~\cite{CoxWang2013root}.  Our result instead gives an exact comparison
with its local volatility projection, calibrated to the entire family of
one-dimensional marginals: the variance-call inequality is strict at every positive
expiry and every positive strike.

Related convex-order inversions for squared VIX are investigated by
Guyon~\cite{Guyon2020} and established in explicitly constructed
stochastic volatility models by Acciaio and Guyon~\cite{AcciaioGuyon2020}.
These comparisons concern a conditional expectation of future integrated
variance and do not by themselves imply an ordering of realized integrated
variance.  In numerical examples for Bergomi models, Guyon finds lower prices
in the associated local volatility model for calls on forward realized
variance~\cite[Section~6.1.2 and Remark~9]{Guyon2020}.
Here we prove the opposite realized-variance ordering over $[0,T]$ in the
classical Heston model with $\rho=-1$, for arbitrary positive parameters,
with strict call inequalities for every $T>0$ and $K>0$.

Completing the arguments outlined above requires considerable technical work.
The necessary well-posedness, regularity, approximation, localization, and
support results are established in the subsequent sections.
In Section~\ref{sec:models} we state the coefficient regularity and uniqueness results and
prove weak existence for the projected Heston SDE.  Section~\ref{sec:comparison}
establishes our main result, the strict convex ordering.
Appendix~\ref{app:densities} constructs the joint law of the variance and
the integrated variance, which is used for the local-variance surface and
for strict comparison.
Appendix~\ref{app:uniqueness} proves uniqueness
in law for the projected Heston SDE.
Appendix~\ref{app:comparison} collects the technical results
needed for the main theorems.
The present note treats only $\rho=-1$; the cases $\rho=+1$ and
$|\rho|<1$ are topics of current research.

\section{The two models}
\label{sec:models}

We work under a risk-neutral measure, with zero interest rates and zero
dividends.  The stochastic volatility model of Heston~\cite{Heston1993} is
given by
\begin{align}
  \frac{\dd S_t}{S_t}&=\sqrt{v_t}\,\dd W_t, \label{eq:heston-S}\\
  \dd v_t&=\kappa(\theta-v_t)\,\dd t
    +\xi\sqrt{v_t}\,\dd Z_t, \label{eq:heston-v}\\
  \dd\langle W,Z\rangle_t&=\rho\,\dd t, \notag
\end{align}
where $S_0>0$, $v_0>0$, $\kappa,\theta,\xi>0$, and $\rho\in[-1,1]$.

Define the Heston integrated variance by
\begin{equation}
  I_t^{\mathrm H}:=\int_0^t v_u\,\dd u.
  \label{eq:heston-integrated-variance}
\end{equation}
Standard localization, the Burkholder--Davis--Gundy inequality and
Gronwall's lemma applied to \eqref{eq:heston-v} give, for every finite
$T>0$ and $p\geq1$,
\begin{equation}
  \E\!\left[\sup_{0\leq t\leq T}v_t^p\right]
    +\E[(I_T^{\mathrm H})^p]<\infty.
  \label{eq:cir-moments}
\end{equation}

For the main result, set
\begin{equation*}
  \rho=-1,
\end{equation*}
so going forward we take $Z=-W$.
Define the log-spot
\[
  L_t:=\log S_t,
\]
so that
\begin{equation}
  L_t-\log S_0=-\frac12 I_t^{\mathrm H}
    +\int_0^t\sqrt{v_u}\,\dd W_u.
  \label{eq:logS}
\end{equation}
Since $Z=-W$, integration of \eqref{eq:heston-v}
gives
\begin{equation}
  v_t-v_0=\kappa\theta t-\kappa I_t^{\mathrm H}
    -\xi\int_0^t\sqrt{v_u}\,\dd W_u.
  \label{eq:int-cir}
\end{equation}
Eliminating the stochastic integral between \eqref{eq:logS} and
\eqref{eq:int-cir}, we obtain
\begin{equation}
  v_t=v_0+\kappa\theta t-\xi\log(S_t/S_0)
    -\left(\kappa+\frac{\xi}{2}\right)I_t^{\mathrm H}.
  \label{eq:pathwise}
\end{equation}
This identity underlies the Heston--Nandi extremality argument in
Cox--Wang~\cite[Section~6, Eq.~(6.2)]{CoxWang2013bounds};
its general-correlation counterpart is the basis for the Monte
Carlo simulation scheme of Broadie and
Kaya~\cite[Section~3, Eqs.~(6)--(7) and~(23)]{BroadieKaya2006}.

For later use, write
\begin{equation}
  \beta:=\kappa+\frac{\xi}{2}>0,
  \qquad A(t,s):=v_0+\kappa\theta t-\xi\log(s/S_0).
  \label{eq:beta-A}
\end{equation}
Then $v_t=A(t,S_t)-\beta I_t^{\mathrm H}$.  Since $v_t\geq0$ and
$I_t^{\mathrm H}>0$ for $t>0$, the Heston spot satisfies
$A(t,S_t)=v_t+\beta I_t^{\mathrm H}>0$ almost surely: at every positive
time it lies strictly below the moving barrier
\begin{equation}
  s^*(t):=S_0\,e^{(v_0+\kappa\theta t)/\xi}.
  \label{eq:spot-barrier}
\end{equation}
On each finite horizon, $s^*(t)\leq s^*(T)$, so $S$ is a bounded local
martingale and hence a true martingale.

We now introduce the corresponding Markovian projection
local volatility model.  Let
\begin{equation}
  \lambda(t,s):=\E[v_t\mid S_t=s],\qquad t>0,\quad s>0.
  \label{eq:lambda}
\end{equation}
The following proposition specifies the version of the conditional
expectation used throughout.
The proof is given in Appendix~\ref{app:densities}.

\begin{proposition}\label{prop:surface}
There exists a nonnegative function
$\lambda:[0,\infty)\times(0,\infty)\to[0,\infty)$ with the following
properties.
\begin{enumerate}[label=\textup{(\roman*)}]
\item For every $t>0$, $\lambda(t,\cdot)$ is a version of
$\E[v_t\mid S_t=\cdot]$.\label{lb:12}
\item $\lambda$ is continuous on
$[0,\infty)\times(0,\infty)$, and
\begin{equation*}
 \lambda(0,s)=\bigl(v_0-\xi\log(s/S_0)\bigr)^+.
\end{equation*}
\item For every $t\geq0$, $\lambda(t,s)>0$ for $0<s<s^*(t)$ and
$\lambda(t,s)=0$ for $s\geq s^*(t)$.
\item For each finite $T>0$ there is
a constant $C_\lambda(T)<\infty$ such that
\begin{equation*}
 \begin{gathered}
 0\leq\lambda(t,s)\leq A(t,s)^+,
 \qquad 0\leq\lambda(t,e^x)\leq C_\lambda(T)(1+|x|),\\
 t\in[0,T],\qquad s>0,\qquad x\in\R.
 \end{gathered}
\end{equation*}
\end{enumerate}
\end{proposition}

The associated Markovian projection
local volatility model is
\begin{equation}
  \frac{\dd\bar S_t}{\bar S_t}
     =\sqrt{\lambda(t,\bar S_t)}\,\dd B_t,
  \qquad \bar S_0=S_0.                              \label{eq:lv-S}
\end{equation}

The next lemma establishes the existence of a weak solution of \eqref{eq:lv-S}, and justifies
its ``Markovian projection'' interpretation.
\begin{lemma}\label{lem:mimicking-existence}
For $\lambda$ as in Proposition~\ref{prop:surface} and every finite
$T>0$, the SDE \eqref{eq:lv-S} has a continuous strictly positive weak
solution $\bar S$ on $[0,T]$ whose one-dimensional marginals agree with
those of $S$.
\end{lemma}

\begin{proof}
Proposition~\ref{prop:surface} and \eqref{eq:cir-moments} allow us to apply
Brunick--Shreve~\cite[Corollary~3.7]{BrunickShreve2013}
to $L=\log S$.  Write $\mu(t,x):=\lambda(t,e^x)$.
In view of Proposition~\ref{prop:surface}\textup{\ref{lb:12}}, the conditional drift and covariance are taken as
$-\mu/2$ and $\mu$, respectively.
This gives a continuous weak solution
\[
 \dd\bar L_t=-\tfrac12\mu(t,\bar L_t)\,\dd t
                +\sqrt{\mu(t,\bar L_t)}\,\dd B_t,
 \qquad \bar L_0=\log S_0,
\]
with the same one-dimensional marginals as $L$.  By It\^o's formula,
$\bar S=e^{\bar L}$ is a strictly positive solution of \eqref{eq:lv-S}, with
\begin{equation*}
 \Law(\bar S_t)=\Law(S_t),\qquad 0\leq t\leq T.
\end{equation*}
\end{proof}

Define the accumulated variance of the projected model by
\[
 I_t^{\mathrm{LV}}:=\int_0^t\lambda(u,\bar S_u)\,\dd u.
\]
We now complete the existence-uniqueness argument for \eqref{eq:lv-S}, relying on regularity
of $\lambda$ established in Proposition~\ref{prop:surface}. The proof is given in Appendix~\ref{app:uniqueness}.
\begin{theorem}\label{thm:timezero}
For every finite $T>0$, uniqueness in law holds for \eqref{eq:lv-S}
on $[0,T]$ among continuous strictly positive weak solutions.
The solution has the Heston one-dimensional marginals and remains strictly
below $s^*(t)$ on $[0,T]$, almost surely.
\end{theorem}

With the projected model now well posed, we turn to an identity that
relates its local variance to $I_t^{\mathrm H}$ and will
be used in the comparison argument.  Taking the conditional expectation
of \eqref{eq:pathwise} given $S_t$ gives
\begin{equation*}
  \lambda(t,S_t)
     =A(t,S_t)-\beta\E[I_t^{\mathrm H}\mid S_t].
\end{equation*}
Fix the Borel version
\begin{equation}
 \bar I(t,s):=\frac{A(t,s)-\lambda(t,s)}{\beta},
 \qquad t\geq0,\quad s>0.
 \label{eq:conditional-mean-integral}
\end{equation}
Thus $\bar I(t,S_t)=\E[I_t^{\mathrm H}\mid S_t]$ almost surely for every
$t\geq0$.
Consequently, we have
\begin{equation}
  v_t-\lambda(t,S_t)
   =-\beta\bigl(I_t^{\mathrm H}-\bar I(t,S_t)\bigr),
   \qquad t\in[0,T].
                                                        \label{eq:centered-relation}
\end{equation}

\begin{remark*}
Ewald~\cite[Theorem~5.1]{Ewald2005} obtains a Malliavin-calculus
representation of the Heston conditional variance under the Feller
condition and for $|\rho|<1$.
Atlan~\cite[Section~1.3.2, Eq.~(1.20)]{Atlan2007} derives the
conditional-variance formula \eqref{eq:lambda}, explains its connection
with Gy\"ongy's mimicking theorem, and develops projection calculations
using Bessel processes and their relation to CIR variance in
\cite[Section~1.4]{Atlan2007}. 
His explicit correlated representation
\cite[Section~1.4.2, Eq.~(1.53)]{Atlan2007} involves division by
$\sqrt{1-\rho^2}$ and therefore cannot be evaluated directly for the
Heston--Nandi model with $\rho=-1$.

\end{remark*}

\section{The comparison identity and convex order}
\label{sec:comparison}

Throughout this section, $T>0$ is fixed. The next theorem presents the first part of our
main result, namely the non-strict convex-order inequality for the Heston--Nandi model and its local volatility projection.

\begin{theorem}\label{thm:main}
For the Heston--Nandi model of Section~\ref{sec:models} and every $T>0$,
\begin{equation}
 \E[(I_T^{\mathrm H}-K)^+]
 \leq\E[(I_T^{\mathrm{LV}}-K)^+],\qquad K\in\R.
 \label{eq:call-ineq}
\end{equation}
Moreover,
\begin{equation*}
 I_T^{\mathrm H}\preceq_{\mathrm{cx}}I_T^{\mathrm{LV}}.
\end{equation*}
\end{theorem}

\begin{proof}
Fix $K\in\R$ and $m\geq1$, and put
\begin{equation*}
 \Phi(i):=\Phi_{m,K}(i):=\frac1m\log(1+e^{m(i-K)}).
\end{equation*}
This function is smooth, nonnegative, nondecreasing, convex and
$1$-Lipschitz, with bounded derivatives of every order.
It approximates the call payoff uniformly from above,
\begin{equation}
  0\leq\Phi_{m,K}(i)-(i-K)^+\leq\frac{\log 2}{m},
  \qquad i\in\R.
  \label{eq:soft-error}
\end{equation}
Proposition~\ref{prop:surface}, Lemma~\ref{lem:mimicking-existence} and
Theorem~\ref{thm:timezero} allow us to apply Lemma~\ref{lem:stability}
in log-spot.  Let $\nu_k$ be the smooth majorants of $\mu(t,x)=\lambda(t,e^x)$
supplied by
Lemma~\ref{lem:stability}, and set
$\lambda_k(t,s):=\nu_k(t,\log s)$.  Then
\begin{equation}
 \begin{gathered}
  \lambda(t,s)+\frac1k\leq\lambda_k(t,s)
       \leq C_0(1+|\log s|),\qquad t\in[0,T],\ s>0,\ k=1,2,\ldots, \\
  \sup_{(t,s)\in[0,T]\times[a,b]}|\lambda_k(t,s)-\lambda(t,s)|
       \longrightarrow0,\qquad 0<a<b<\infty.
 \end{gathered}
 \label{eq:majorants}
\end{equation}
Lemma~\ref{lem:stability} ensures that the approximating local volatility SDEs
\begin{equation}
 \frac{\dd S_u^k}{S_u^k}
       =\sqrt{\lambda_k(u,S_u^k)}\,\dd B_u,
 \label{eq:approximating-sde}
\end{equation}
have unique strictly positive strong solutions
on $[t,T]$, for every deterministic starting time $t\in[0,T]$ and
initial value $S_t^k=s>0$.
Define
\begin{equation}
 \begin{aligned}
 V_k(t,s,i)&:=\E_{t,s}\!\left[
  \Phi\!\left(i+\int_t^T\lambda_k(u,S_u^k)\,\dd u\right)\right],\\
 H_k(t,s,i)&:=\tfrac12s^2\partial_{ss}V_k(t,s,i)
                         +\partial_iV_k(t,s,i).
 \end{aligned}
 \label{eq:Uk-hk}
\end{equation}
Take $\nu=\nu_k$
and $\psi=\Phi$ in Lemma~\ref{lem:smooth}.  The chain rule then gives
\[
 V_k(t,e^x,i)=U(t,x,i),\qquad H_k(t,e^x,i)=h(t,x,i),
\]
and consequently,
\[
 \partial_tV_k+\lambda_k H_k=0,\qquad V_k(T,s,i)=\Phi(i).
\]
By Lemma~\ref{lem:smooth} and the identities
$e^x\partial_sV_k(t,e^x,i)=\partial_xU(t,x,i)$,
$H_k(t,e^x,i)=\Dop U(t,x,i)$, both functions are bounded for each fixed $k$, and
\begin{equation}
 |V_k(t,s,i)|\leq c_k(1+|\log s|+i),\qquad i\geq0.
 \label{eq:Vk-growth-bound}
\end{equation}
The Feynman--Kac representation in Lemma~\ref{lem:smooth} also gives
$H_k\geq0$ and shows that $i\mapsto H_k(t,s,i)$ is nondecreasing.

For $n\geq1$, define the stopping time
\[
 \varsigma_n:=\inf\bigl\{u\in[0,T]:
          S_u\notin(1/n,n)\ \text{or}\ v_u+I_u^{\mathrm H}\geq n\bigr\},
 \qquad \inf\varnothing:=T.
\]
Continuity of $S$ and $v$, together with strict positivity of $S$,
gives $\varsigma_n\uparrow T$
almost surely.  Apply It\^o's formula to
$V_k(t,S_t,I_t^{\mathrm H})$, using
\eqref{eq:heston-S} and \eqref{eq:heston-integrated-variance},
on $[0,\varsigma_n]$ to obtain
\begin{equation}
 \begin{aligned}
 &V_k(\varsigma_n,S_{\varsigma_n},I_{\varsigma_n}^{\mathrm H})
       -V_k(0,S_0,0)\\
 &\quad=\int_0^{\varsigma_n}
      \bigl(v_t-\lambda_k(t,S_t)\bigr)H_k(t,S_t,I_t^{\mathrm H})\,\dd t\\
 &\qquad\quad+\int_0^{\varsigma_n}
      S_t\sqrt{v_t}\,\partial_sV_k(t,S_t,I_t^{\mathrm H})\,\dd W_t.
 \end{aligned}
 \label{eq:stopped-ito}
\end{equation}
Here we used the backward equation $\partial_tV_k=-\lambda_k H_k$.
The pathwise identity \eqref{eq:pathwise} and \eqref{eq:majorants} give
$\lambda_k(t,S_t)\leq C_T(1+v_t+I_t^{\mathrm H})$.
Together with boundedness of $s\partial_sV_k$ and $H_k$, the Heston
moment bounds \eqref{eq:cir-moments} therefore imply
\[
 \E\int_0^T\!\left[
  v_tS_t^2|\partial_sV_k(t,S_t,I_t^{\mathrm H})|^2
  +(v_t+\lambda_k(t,S_t))|H_k(t,S_t,I_t^{\mathrm H})|
 \right]\dd t<\infty,
\]
ensuring that the stopped stochastic integral in \eqref{eq:stopped-ito} converges
in $L^2$, and the drift integral converges in $L^1$, as $n\to\infty$.
For the terminal value, \eqref{eq:Vk-growth-bound} and \eqref{eq:pathwise} give,
for some constant $C_{k,T}>0$,
\[
 \bigl|V_k(\varsigma_n,S_{\varsigma_n},
                         I_{\varsigma_n}^{\mathrm H})\bigr|
 \leq C_{k,T}\bigl(1+\sup_{0\leq u\leq T}v_u+I_T^{\mathrm H}\bigr).
\]
The right-hand side is integrable by \eqref{eq:cir-moments}.
Continuity of $V_k$ and $V_k(T,s,i)=\Phi(i)$ therefore give $L^1$
convergence of the stopped
terminal value to $\Phi(I_T^{\mathrm H})$ as $n\to\infty$.
Taking expectations yields
\begin{equation}
 \E[\Phi(I_T^{\mathrm H})]-V_k(0,S_0,0)
 =\int_0^T\E\!\left[(v_t-\lambda_k(t,S_t))
                   H_k(t,S_t,I_t^{\mathrm H})\right]\dd t.
 \label{eq:majorant-duhamel}
\end{equation}

With $\bar I$ defined by \eqref{eq:conditional-mean-integral}, set
\[
 C_k(t):=\E\!\left[(I_t^{\mathrm H}-\bar I(t,S_t))
                     H_k(t,S_t,I_t^{\mathrm H})\right],\qquad t>0.
\]
Let $\eta_{t,s}$ be a regular conditional law of $I_t^{\mathrm H}$ given
$S_t=s$.  Since $\bar I(t,s)=\int_{\R_+}j\,\eta_{t,s}(\dd j)$,
we have, $\Law(S_t)$-a.s.,
\begin{equation}
 \begin{aligned}
 &\E\!\left[(I_t^{\mathrm H}-\bar I(t,s))
       H_k(t,s,I_t^{\mathrm H})\mid S_t=s\right]\\
 &\quad=\frac12\int_{\R_+^2}(i-j)
       \bigl(H_k(t,s,i)-H_k(t,s,j)\bigr)
       \,\eta_{t,s}(\dd i)\eta_{t,s}(\dd j)\geq0.
 \end{aligned}
\label{eq:chebyshev}
\end{equation}
The integrand is nonnegative because $i\mapsto H_k(t,s,i)$ is
nondecreasing
(see, for example, Stroock~\cite[Exercise~1.1.9(i)]{Stroock2024}),
so that integrating over $s$ gives $C_k(t)\geq0$.  Thus, by
\eqref{eq:centered-relation}, the first inequality in \eqref{eq:majorants},
and the nonnegativity of $H_k$ (ensured by Lemma~\ref{lem:smooth}),
\[
 \begin{aligned}
 &\E[(v_t-\lambda_k(t,S_t))H_k(t,S_t,I_t^{\mathrm H})]\\
 &\quad=-\beta C_k(t)
   -\E[(\lambda_k(t,S_t)-\lambda(t,S_t))H_k(t,S_t,I_t^{\mathrm H})]
   \leq0.
 \end{aligned}
\]
Consequently,
\begin{equation}
 \E[\Phi(I_T^{\mathrm H})]\leq V_k(0,S_0,0).
 \label{eq:majorant-bound}
\end{equation}

Let $I_t^k:=\int_0^t\lambda_k(u,S_u^k)\dd u$, with $S_0^k=S_0$.
Lemma~\ref{lem:stability} and the continuous mapping theorem give
$(S^k,I^k)\Rightarrow(\bar S,I^{\mathrm{LV}})$.  The uniform fourth-moment bound
\eqref{eq:approx-integral-fourth-moment} and the linear growth of $\Phi$ imply
\begin{equation}
 V_k(0,S_0,0)=\E[\Phi(I_T^k)]
       \longrightarrow\E[\Phi(I_T^{\mathrm{LV}})].
 \label{eq:majorant-value-limit}
\end{equation}
Passing to the limit in \eqref{eq:majorant-bound} proves
\eqref{eq:call-ineq} with $\Phi_{m,K}$ in place of the call payoff.
By \eqref{eq:soft-error}, letting $m\to\infty$
gives \eqref{eq:call-ineq}.
Finally, using \eqref{eq:lambda}, Tonelli's theorem and mimicking, we obtain
\begin{align}
 \E[I_T^{\mathrm H}]
 &=\int_0^T\E[v_t]\,\dd t
   =\int_0^T\E[\lambda(t,S_t)]\,\dd t\notag\\
 &=\int_0^T\E[\lambda(t,\bar S_t)]\,\dd t
   =\E[I_T^{\mathrm{LV}}]<\infty.
 \label{eq:equal-means}
\end{align}
Therefore, the Carr--Madan formula
\cite[p.~134, equation~(11.1)]{Gatheral2006} gives
$I_T^{\mathrm H}\preceq_{\mathrm{cx}}I_T^{\mathrm{LV}}$.
\end{proof}

\subsection{Strict convex order}

The strict comparison uses the following consequence of the joint-law
construction in Appendix~\ref{app:densities}.

\begin{lemma}\label{lem:strict-density}
Let $\beta$ and $A$ be as in \eqref{eq:beta-A}, and let $s^*$ be the
barrier defined in \eqref{eq:spot-barrier}.  Write $p(t,\cdot)$,
$f(t,\cdot,\cdot)$ and $g(t,\cdot)$
for the densities defined in \eqref{eq:cir-density},
\eqref{eq:bridge-convolution} and \eqref{eq:gap-densities}, respectively.
Put
\[
 \mathcal O:=\{(t,s,i):t>0,\ 0<s<s^*(t),\ 0<i<A(t,s)/\beta\}.
\]
\begin{enumerate}[label=\textup{(\roman*)}]
\item For every $t>0$, the law of $(S_t,I_t^{\mathrm H})$ has the density
\begin{equation}
 \varrho(t,s,i)=\frac{\xi}{s}\,
 p(t,A(t,s)-\beta i)f(t,A(t,s)-\beta i,i),\qquad (t,s,i)\in\mathcal O,
 \label{eq:rho-density}
\end{equation}
extended by zero outside $\mathcal O$.
The function $\varrho$ is continuous and strictly positive on
$\mathcal O$.

\item For every $t>0$, the spot density is
\begin{equation}
 f^S(t,s)=\frac{\xi}{s}g(t,A(t,s)),\qquad 0<s<s^*(t).
 \label{eq:spot-density}
\end{equation}
The map $(t,s)\mapsto f^S(t,s)$ is continuous and strictly
positive for $t>0$ and $0<s<s^*(t)$.

\item The formula
\begin{equation}
 \eta_{t,s}(\dd i):=\frac{\varrho(t,s,i)}{f^S(t,s)}\,\dd i,
 \qquad t>0,\quad 0<s<s^*(t),
 \label{eq:conditional-density}
\end{equation}
defines a regular conditional law
of $I_t^{\mathrm H}$ given $S_t=s$.
\end{enumerate}
\end{lemma}

\begin{proof}
For \textup{(i)}, Lemma~\ref{lem:cir-joint} gives the density
$p(t,v)f(t,v,i)$.  The change of variables
$(s,i)\mapsto(A(t,s)-\beta i,i)$ has absolute Jacobian $\xi/s$,
which gives \eqref{eq:rho-density}.  Its continuity and positivity hold
on $\mathcal O$, where both arguments are strictly positive.
For \textup{(ii)}, integrating \eqref{eq:rho-density} in $i$ and using
\eqref{eq:gap-densities} gives \eqref{eq:spot-density};
Proposition~\ref{prop:gap-densities} gives its continuity and positivity.
For \textup{(iii)}, the density ratio \eqref{eq:conditional-density}
integrates to one and
satisfies the defining property of a regular conditional distribution.
To define $\eta_{t,s}$
for every $s\in\R$, $t>0$, set $\eta_{t,s}:=\delta_0$ when
$s\notin(0,s^*(t))$, where $\delta_0$ is the Dirac measure at $i=0$.
Since $\mathbb P(0<S_t<s^*(t))=1$, this extension outside
$(0,s^*(t))$ does not change the conditional identity.
The density ratio is Borel measurable on its domain, and the region
$\{(t,s):t>0,\ 0<s<s^*(t)\}$ is Borel.  Integration in $i$ and the
constant extension therefore show that, for every Borel set
$E\subset\R_+$, the map $(t,s)\mapsto\eta_{t,s}(E)$ is Borel measurable
on $(0,\infty)\times\R$.
\end{proof}

\begin{theorem}\label{thm:strict}
For the Heston--Nandi model of Section~\ref{sec:models}, for every $T>0$ and every $K>0$,
\begin{equation}
 \E[(I_T^{\mathrm H}-K)^+]<\E[(I_T^{\mathrm{LV}}-K)^+].
 \label{eq:strict-call}
\end{equation}
Consequently,
\begin{equation}
 I_T^{\mathrm H}\prec_{\mathrm{cx}}I_T^{\mathrm{LV}}.
 \label{eq:strict-cx}
\end{equation}
\end{theorem}

\begin{proof}
Fix $T>0$ and $K>0$, and set
\[
 D(K):=\E[(I_T^{\mathrm{LV}}-K)^+]-\E[(I_T^{\mathrm H}-K)^+],
\]
which is nonnegative by Theorem~\ref{thm:main}.  We prove
\eqref{eq:strict-call} by bounding $D(K)$ from below by a positive
constant.

\smallskip\noindent
\emph{Step 1: connect the comparison gap to increments of $H_n$.}
Let $\lambda_n$ be the coefficients in \eqref{eq:majorants}, and let
$V_n,H_n$ denote \eqref{eq:Uk-hk} with coefficient $\lambda_n$ and payoff
$\Phi_n:=\Phi_{n,K}$.  By Lemma~\ref{lem:stability} and
\eqref{eq:soft-error}, exactly as in \eqref{eq:majorant-value-limit},
\begin{equation}
 D_n:=V_n(0,S_0,0)-\E[\Phi_n(I_T^{\mathrm H})]\longrightarrow D(K).
 \label{eq:strict-gap-limit}
\end{equation}
By \eqref{eq:majorant-duhamel} and \eqref{eq:centered-relation}, together
with $\lambda_n\geq\lambda$ and $H_n\geq0$,
\begin{equation}
 D_n\geq\beta\int_0^T C_n(t)\,\dd t,\qquad
 C_n(t):=\E\bigl[(I_t^{\mathrm H}-\bar I(t,S_t))
                  H_n(t,S_t,I_t^{\mathrm H})\bigr].
 \label{eq:strict-majorant-gap}
\end{equation}
Fix $t_0\in(0,T)$, $0<s_1<s_2<s^*(t_0)$ and
$0<i_1<i_2<K<i_3<i_4<A(t_0,s_2)/\beta$, which is possible since
$A(t_0,s)\to\infty$ as $s\downarrow0$.  The compact set
$\mathcal B:=[t_0,T]\times[s_1,s_2]\times[i_1,i_4]$ then lies in the
set $\mathcal O$ defined in Lemma~\ref{lem:strict-density}, so that, with $\varrho$ and $f^S$
given by \eqref{eq:rho-density} and \eqref{eq:spot-density}, respectively, we have
\[
 \varrho_{\min}:=\min_{\mathcal B}\varrho>0,
 \qquad
 f_{\max}:=\max_{[t_0,T]\times[s_1,s_2]}f^S<\infty.
\]
Put $G_n(t,s):=H_n(t,s,i_3)-H_n(t,s,i_2)$, which is nonnegative because
$H_n$ is nondecreasing in $i$.  Integrating \eqref{eq:chebyshev} against
the spot density $f^S(t,\cdot)$ expresses $C_n(t)$ as an integral of a
nonnegative function of $(s,i,j)$.  Keeping only
$s\in[s_1,s_2]$ and $(i,j)\in[i_3,i_4]\times[i_1,i_2]$ together with the
symmetric region (which cancels the factor $\tfrac12$), and using
$i-j\geq i_3-i_2$, $H_n(t,s,i)-H_n(t,s,j)\geq G_n(t,s)$ and
$f^S(t,s)\eta_{t,s}(\dd i)=\varrho(t,s,i)\,\dd i$, we obtain for
$t\in[t_0,T]$
\begin{equation}
 C_n(t)\geq c\int_{s_1}^{s_2}G_n(t,s)\,\dd s,
 \qquad
 c:=\frac{(i_3-i_2)(i_4-i_3)(i_2-i_1)\,\varrho_{\min}^2}{f_{\max}}>0.
 \label{eq:strict-density-bound}
\end{equation}
By \eqref{eq:majorants},
$\Lambda:=\sup_n\max_{[t_0,T]\times[s_1,s_2]}\lambda_n<\infty$.  Since
$C_n\geq0$ by \eqref{eq:chebyshev} and $G_n\geq0$,
\eqref{eq:strict-majorant-gap} and \eqref{eq:strict-density-bound} give
\begin{equation}
 D_n\geq\beta c\int_{s_1}^{s_2}\int_{t_0}^T G_n(t,s)\,\dd t\,\dd s
 \geq\frac{\beta c}{\Lambda}\int_{s_1}^{s_2}R_n(s)\,\dd s,
 \qquad
 R_n(s):=\int_{t_0}^T\lambda_n(t,s)G_n(t,s)\,\dd t.
 \label{eq:strict-increment-bound}
\end{equation}

\smallskip\noindent
\emph{Step 2: Relate $R_n$ to expected truncated future variance.}
For $s\in[s_1,s_2]$ let $S^n$ be the solution of
\eqref{eq:approximating-sde} started at $s$ at time $t_0$, and put
$J_n(s):=\int_{t_0}^T\lambda_n(u,S_u^n)\,\dd u$.  Integrating
$\partial_tV_n=-\lambda_nH_n$ over $[t_0,T]$ at fixed $(s,i)$ and using
$V_n(T,s,i)=\Phi_n(i)$ gives
$\int_{t_0}^T\lambda_n(t,s)H_n(t,s,i)\,\dd t
 =\E^n_{t_0,s}[\Phi_n(i+J_n(s))]-\Phi_n(i)$, hence
\[
 R_n(s)=\E^n_{t_0,s}[\Delta_n(J_n(s))]-\Delta_n(0),
 \qquad
 \Delta_n(z):=\Phi_n(i_3+z)-\Phi_n(i_2+z).
\]
With $\Delta(z):=(i_3+z-K)^+-(i_2+z-K)^+$ we have
$\Delta(z)-\Delta(0)=z\wedge(K-i_2)$ for $z\geq0$, because $i_2<K<i_3$,
and $|\Delta_n-\Delta|\leq(\log2)/n$ by \eqref{eq:soft-error}.  Therefore
\begin{equation}
 R_n(s)\geq\E^n_{t_0,s}\bigl[J_n(s)\wedge(K-i_2)\bigr]-\frac{2\log2}{n},
 \qquad s\in[s_1,s_2].
 \label{eq:strict-truncated-variance}
\end{equation}

\smallskip\noindent
\emph{Step 3: the future variance does not vanish.}
Fix $\delta>0$ with $s_2e^{\delta}<s^*(t_0)$ and let
$\mathcal K:=[t_0,T]\times[s_1e^{-\delta},s_2e^{\delta}]$.  Since $s^*$ is
increasing, $\mathcal K$ lies strictly below the barrier, so
Proposition~\ref{prop:surface} gives $m:=\min_{\mathcal K}\lambda>0$, and
the first inequality in \eqref{eq:majorants} gives
$\lambda_n\geq m$ on $\mathcal K$ for every $n$.  Set
\[
 \eta:=\frac\delta2,\qquad
 \varepsilon:=\min\Bigl\{\delta,\ \frac{\eta^2}{2},\ \frac{(T-t_0)m}{2}\Bigr\}.
\]
For $s\in[s_1,s_2]$ write, for $u\in[t_0,T]$,
\[
 \log S_u^n-\log s=-\tfrac12Q_u^n+M_u^n,
 \qquad
 Q_u^n:=\int_{t_0}^u\lambda_n(r,S_r^n)\,\dd r=\langle M^n\rangle_u,
\]
so that $Q_T^n=J_n(s)$.  On the event
$\{J_n(s)<\varepsilon\}\cap\{\sup_{u}|M_u^n|\leq\eta\}$ we have
$\sup_u|\log S_u^n-\log s|\leq\varepsilon/2+\eta\leq\delta$, so
$(u,S_u^n)\in\mathcal K$ for all $u\in[t_0,T]$ and hence
$J_n(s)\geq(T-t_0)m\geq2\varepsilon$; the event is therefore empty.
Let $\tau:=\inf\{u\in[t_0,T]:Q_u^n\geq\varepsilon\}\wedge T$,
with $\inf\varnothing=\infty$.  Then
$\langle M^n\rangle_\tau\leq\varepsilon$, and $\tau=T$ on
$\{J_n(s)<\varepsilon\}$, so Doob's maximal inequality gives
\[
 \mathbb P^n_{t_0,s}\bigl(J_n(s)<\varepsilon\bigr)
 \leq\mathbb P^n_{t_0,s}\Bigl(\sup_{t_0\leq u\leq\tau}|M_u^n|>\eta\Bigr)
 \leq\frac{\E^n_{t_0,s}[(M_\tau^n)^2]}{\eta^2}
 \leq\frac{\varepsilon}{\eta^2}\leq\frac12.
\]
Consequently, for every $n$ and every $s\in[s_1,s_2]$,
\begin{equation}
 \E^n_{t_0,s}\bigl[J_n(s)\wedge(K-i_2)\bigr]
 \geq\tfrac12\bigl(\varepsilon\wedge(K-i_2)\bigr)=:\gamma>0.
 \label{eq:strict-variance-floor}
\end{equation}

\smallskip\noindent
\emph{Step 4: the comparison gap is strictly positive.}
Combining \eqref{eq:strict-increment-bound},
\eqref{eq:strict-truncated-variance} and \eqref{eq:strict-variance-floor},
\[
 D_n\geq\frac{\beta c}{\Lambda}(s_2-s_1)\Bigl(\gamma-\frac{2\log2}{n}\Bigr),
\]
and letting $n\to\infty$ in \eqref{eq:strict-gap-limit} yields
\[
 D(K)\geq\frac{\beta c\,(s_2-s_1)\,\gamma}{\Lambda}>0,
\]
which is \eqref{eq:strict-call}.
For $K\leq0$, the two call values equal $\E[I_T^{\mathrm H}]-K$ by
nonnegativity and \eqref{eq:equal-means}, so the positive-strike
restriction is sharp.  Finally, \eqref{eq:strict-cx} follows from
\eqref{eq:strict-call} and Theorem~\ref{thm:main}.
\end{proof}

The preceding results yield the following consequences.
\begin{corollary}\label{cor:consequences}
For every $T>0$ and $p\geq1$,
\[
                    \E[(I_T^{\mathrm{LV}})^p]<\infty.
\]
In particular,
\begin{equation}
 \operatorname{Var}(I_T^{\mathrm H})
 <\operatorname{Var}(I_T^{\mathrm{LV}}),                  \label{eq:variance-consequence}
\end{equation}
and the fair strikes of the corresponding continuously monitored volatility
swaps satisfy
\begin{equation}
 \E\!\left[\sqrt{I_T^{\mathrm H}/T}\right]
 >\E\!\left[\sqrt{I_T^{\mathrm{LV}}/T}\right].           \label{eq:volswap-consequence}
\end{equation}
\end{corollary}

\noindent We omit the proof, which follows from the moment bounds and the strict call inequalities.


Thus at $\rho=-1$ the volatility-swap ordering is the reverse of the
practitioner ordering mentioned in the introduction.

For comparison, the mimicking property gives
$\Law(\lambda(t,\bar S_t))=\Law(\lambda(t,S_t))$, and conditional Jensen gives
\begin{equation*}
 \lambda(t,S_t)=\E[v_t\mid S_t]\preceq_{\mathrm{cx}}v_t.
\end{equation*}
Therefore, the instantaneous variance in the local volatility model is smaller in
convex order, whereas its integrated variance is larger.

\appendix

\input{final_appendix_densities}

\input{final_appendix_uniqueness}

\input{final_appendix_comparison}

\ifincludenumerics
\input{numerics_section}
\clearpage
\fi

\section*{Declaration on the use of generative AI}

During the development and preparation of this paper, the author used
OpenAI's ChatGPT and Anthropic's Claude as exploratory mathematical tools.
The central analytical
ingredients---notably the use of the Duhamel identity and the Chebyshev
monotonicity argument---were introduced by the author in his interactions
with the systems.  ChatGPT assisted in combining and developing these
author-supplied ideas, carrying out intermediate calculations, testing the
resulting implications, and identifying the construction that led to the
reverse-ordering result presented in the paper.  Claude was used to verify
the work\ifincludenumerics\space
and to produce the numerical results\fi.
Thus, although the underlying
proof strategy originated with the author, the result emerged from
AI-assisted exploration.  The author is responsible for independently reviewing and verifying
every mathematical argument in the manuscript.  All mathematical judgments and
conclusions remain the responsibility of the
author.

\input{final_bibliography}

\end{document}

%% file: final_appendix_densities.tex
\section{Joint densities and the local variance surface}
\label{app:densities}

In this appendix, we develop a series of technical results leading to
the proof of Proposition~\ref{prop:surface}.

Put
\begin{equation}
 Q_t:=v_t+\beta I_t^{\mathrm H}=A(t,S_t),\qquad
 \alpha:=\frac{2\kappa\theta}{\xi^2}>0.
 \label{eq:gap-process}
\end{equation}
We first construct the joint density of $(v_t,I_t^{\mathrm H})$.
Using the identity
$S_t=S_0\exp\{(v_0+\kappa\theta t-v_t-\beta I_t^{\mathrm H})/\xi\}$,
we then obtain the joint density of $(S_t,I_t^{\mathrm H})$ needed for
the strict comparison in Section~\ref{sec:comparison} and the continuous
local variance function in Proposition~\ref{prop:surface}.

\subsection{The conditional law of integrated variance}
This subsection builds on Glasserman and Kim~\cite{GlassermanKim2011}.
Before stating the technical lemmas needed to establish Proposition~\ref{prop:surface},
we provide a bound for the CIR transition density.

For $w\geq0$ and $\alpha$ as in \eqref{eq:gap-process}, write
\[
 \mathcal I_{\alpha-1}(w):=
 \sum_{j\geq0}\frac{(w/2)^{2j}}{j!\,\Gamma(j+\alpha)}.
\]
Thus $I_{\alpha-1}(w)=(w/2)^{\alpha-1}\mathcal I_{\alpha-1}(w)$ for
$w>0$, where $I_{\alpha-1}$ is the modified Bessel function of the first
kind; see Broadie and Kaya~\cite[Eq.~(21), p.~221]{BroadieKaya2006}.
The CIR transition law is a scaled noncentral chi-square distribution
with $2\alpha$ degrees of freedom
\cite[Eq.~(2.3), p.~269]{GlassermanKim2011}.
Writing
\begin{equation*}
 d(t):=\frac{\xi^2(1-e^{-\kappa t})}{2\kappa},\qquad
 w(t,v):=\frac{2\kappa\sqrt{v_0v}}{\xi^2\sinh(\kappa t/2)},
 \qquad t>0,\quad v\geq0,
\end{equation*}
its density is
\begin{equation}
 p(t,v)=d(t)^{-\alpha}v^{\alpha-1}
 e^{-(v+v_0e^{-\kappa t})/d(t)}\mathcal I_{\alpha-1}(w(t,v)),
 \qquad v>0.
 \label{eq:cir-density}
\end{equation}

The function $p$ is continuous and
strictly positive for $t,v>0$.  For every $0<\tau<T<\infty$ and $V>0$,
there exists a finite constant $C_{\tau,T,V}$ such that
\begin{equation}
 p(t,v)\leq C_{\tau,T,V}v^{\alpha-1},
 \qquad \tau\leq t\leq T,\quad 0<v\leq V.
 \label{eq:cir-boundary-domination}
\end{equation}
Indeed, on this range $d(t)$ is bounded away from zero and $w(t,v)$ is
bounded, hence so is $\mathcal I_{\alpha-1}(w(t,v))$.

The conditional law of the integrated variance $I_t^{\mathrm H}$, given
$v_t=v$, has a gamma-series representation due to
Glasserman and Kim~\cite[Theorem~2.2, pp.~270--271]{GlassermanKim2011}.
We give the relevant details next.  Set
\begin{equation}
 \gamma_n(t):=\frac{\kappa^2t^2+4\pi^2n^2}{2\xi^2t^2},\qquad
 \ell_n(t,v):=\frac{16\pi^2n^2(v_0+v)}
 {\xi^2t(\kappa^2t^2+4\pi^2n^2)},\qquad n\geq1.
 \label{eq:gamma-rates}
\end{equation}
For fixed $t>0$ and $v\geq0$, consider mutually independent variables
\[
 N_n\sim\operatorname{Poisson}(\ell_n(t,v)),\quad
 E_{n,j}\sim\operatorname{Exp}(1),\quad
 G_n\sim\operatorname{Gamma}(\alpha,1),\quad
 G'_{n,j}\sim\operatorname{Gamma}(2,1),
\]
and an independent integer-valued variable $\mathsf B$ with Bessel distribution
\begin{equation}
 \mathbb P(\mathsf B=j)=
 \frac{(w(t,v)/2)^{2j}}
 {j!\,\Gamma(j+\alpha)\mathcal I_{\alpha-1}(w(t,v))},\qquad j\geq0.
 \label{eq:bessel-mixture}
\end{equation}
Here $\operatorname{Poisson}(\ell)$ denotes the Poisson distribution with
mean $\ell$, $\operatorname{Exp}(1)$ the exponential distribution with
rate $1$, and $\operatorname{Gamma}(a,1)$ the gamma
distribution with shape $a$ and rate $1$.  Define
\begin{equation}
 \begin{gathered}
 U_{t,v}:=\Xi_1+\Xi_2+\Xi_3,\qquad
 \Xi_1:=\sum_{n\geq1}\frac1{\gamma_n}\sum_{j=1}^{N_n}E_{n,j},\\
 \Xi_2:=\sum_{n\geq1}\frac{G_n}{\gamma_n},\qquad
 \Xi_3:=\sum_{j=1}^{\mathsf B}\sum_{n\geq1}\frac{G'_{n,j}}{\gamma_n}.
 \end{gathered}
 \label{eq:gamma-expansion}
\end{equation}
Tonelli's theorem and independence give
\begin{equation}
 \begin{aligned}
  \E\Xi_1&=\sum_{n\geq1}\frac{\ell_n(t,v)}{\gamma_n(t)},&
  \E\Xi_2&=\alpha\sum_{n\geq1}\frac1{\gamma_n(t)},\\
  \E\Xi_3&=2\E[\mathsf B]\sum_{n\geq1}\frac1{\gamma_n(t)}.
 \end{aligned}
 \label{eq:gamma-component-means}
\end{equation}
The expectations in \eqref{eq:gamma-component-means} are finite because
$\sum_n\gamma_n^{-1}<\infty$,
$\sup_n\ell_n<\infty$, and $\mathsf B$ has finite moments of all orders;
hence the series in \eqref{eq:gamma-expansion} converge almost surely.

Fix an integer $M\geq1$ with $r:=M\alpha>1$, and write
\begin{equation}
 U_{t,v}=Y_t+\mathcal R_{t,v},\qquad
 Y_t:=\sum_{n=1}^{M}G_n/\gamma_n(t).
 \label{eq:gamma-block}
\end{equation}
The condition $M\alpha>1$ ensures that the density of $Y_t$, which behaves
like $y^{M\alpha-1}$ near zero, has a bounded, continuous extension by zero.
This will be used in the convolution argument below.
Set $\eta_0:=\pi^2/\xi^2$.  Then $Y_t$ and $\mathcal R_{t,v}$ are
independent and nonnegative.

\begin{lemma}\label{lem:bridge-small-time}
For every $V>0$, there are constants $C_1,C_2,D>0$ such that, for
$0<t\leq1$ and $0\leq v\leq V$,
\begin{equation}
 \E e^{\eta_0 U_{t,v}/t^2}\leq C_1 e^{D/t},
 \label{eq:bridge-exp-moment}
\end{equation}
and
\begin{equation}
 \E e^{\eta_0\mathcal R_{t,v}/t^2}\leq C_1 e^{D/t},\qquad
 \E\mathcal R_{t,v}\leq C_2t.
 \label{eq:bridge-remainder-estimates}
\end{equation}
\end{lemma}

\begin{proof}
Fix $V>0$.  There is a constant $C_0(V)>0$ such that, for
$0<t\leq1$, $0\leq v\leq V$, and $n\geq1$,
\begin{equation}
 \frac{c_-n^2}{t^2}\leq\gamma_n(t)\leq\frac{c_+n^2}{t^2},\qquad
 \ell_n(t,v)\leq\frac{C_0(V)}t,\qquad
 w(t,v)\leq\frac{C_0(V)}t,
 \label{eq:gamma-small-time}
\end{equation}
where $c_-=2\pi^2/\xi^2$ and
$c_+=(\kappa^2+4\pi^2)/(2\xi^2)$.

Since $\eta_0=c_-/2$, set
\[
 a_n(t):=\frac{\eta_0}{t^2\gamma_n(t)}\leq\frac1{2n^2},\qquad
 P(t):=\prod_{n\geq1}(1-a_n(t))^{-1}\leq e^{\pi^2/6},
\]
where the last inequality follows from
$-\log(1-a)\leq2a$ for $0\leq a\leq1/2$.
Put $\vartheta:=\eta_0/t^2$, so that
$a_n(t)=\vartheta/\gamma_n(t)<1$.
Computing the moment-generating functions of finite truncations
of the Glasserman--Kim expansion
\cite[Theorem~2.2 and Eq.~(2.8), pp.~270--271]{GlassermanKim2011}, together with
\eqref{eq:bessel-mixture} and monotone convergence then gives
\begin{equation}
 \begin{aligned}
  \E e^{\vartheta\Xi_1}
  &=\exp\!\left\{\sum_{n\geq1}\ell_n(t,v)
       \frac{a_n(t)}{1-a_n(t)}\right\},&
  \E e^{\vartheta\Xi_2}&=P(t)^\alpha,\\
  \E[e^{\vartheta\Xi_3}\mid\mathsf B]&=P(t)^{2\mathsf B},&
  \E e^{\vartheta\Xi_3}
  &=\frac{\mathcal I_{\alpha-1}(P(t)w(t,v))}
          {\mathcal I_{\alpha-1}(w(t,v))}.
 \end{aligned}
 \label{eq:bridge-component-transforms}
\end{equation}
Multiplying the three independent factors in
\eqref{eq:bridge-component-transforms} yields
\[
 \E e^{\eta_0 U_{t,v}/t^2}
 =P(t)^\alpha
   \exp\!\left\{\sum_{n\geq1}\ell_n(t,v)
                  \frac{a_n(t)}{1-a_n(t)}\right\}
   \frac{\mathcal I_{\alpha-1}(P(t)w(t,v))}
        {\mathcal I_{\alpha-1}(w(t,v))}.
\]
The bounds \eqref{eq:gamma-small-time} and
$a_n(t)\leq1/(2n^2)\leq1/2$ imply
\[
 \sum_{n\geq1}\ell_n(t,v)\frac{a_n(t)}{1-a_n(t)}
 \leq\frac{2C_0(V)}{t}\sum_{n\geq1}a_n(t)
 \leq\frac{C_0(V)}{t}\sum_{n\geq1}\frac1{n^2}
 =\frac{C_0(V)\pi^2}{6t}.
\]
Let $d_\alpha:=\min\{1,\alpha\}>0$.  Since
$\Gamma(\alpha+j)\geq\Gamma(\alpha)d_\alpha^j j!$ for $j\geq0$, we obtain
\[
 \frac1{\Gamma(\alpha)}\leq\mathcal I_{\alpha-1}(z)
 \leq\frac{e^{z/\sqrt{d_\alpha}}}{\Gamma(\alpha)},\qquad z\geq0,
\]
where in the last step we used
$\sum_{j\geq0}(u/2)^{2j}/(j!)^2\leq e^u$.
Using $P(t)\leq e^{\pi^2/6}$ and applying the upper bound to the numerator
and the lower bound to the denominator gives
\[
 \frac{\mathcal I_{\alpha-1}(P(t)w(t,v))}
      {\mathcal I_{\alpha-1}(w(t,v))}
 \leq \exp\!\left\{\frac{P(t)w(t,v)}{\sqrt{d_\alpha}}\right\}
 \leq \exp\!\left\{\frac{C_0(V)e^{\pi^2/6}}{\sqrt{d_\alpha}\,t}\right\}.
\]
Combining these estimates yields
\[
 \E e^{\eta_0 U_{t,v}/t^2}
 \leq e^{\alpha\pi^2/6}
       \exp\!\left\{\frac{C_0(V)\pi^2}{6t}\right\}
       \exp\!\left\{\frac{C_0(V)e^{\pi^2/6}}{\sqrt{d_\alpha}\,t}\right\}
 =C_1 e^{D/t},
\]
where $C_1:=e^{\alpha\pi^2/6}$ and
$D:=C_0(V)\bigl(\pi^2/6+e^{\pi^2/6}/\sqrt{d_\alpha}\bigr)$.
This proves \eqref{eq:bridge-exp-moment}.

Since $0\leq\mathcal R_{t,v}\leq U_{t,v}$, the first estimate in
\eqref{eq:bridge-remainder-estimates} follows from
\eqref{eq:bridge-exp-moment}.  Jensen's inequality gives
\[
 \E\mathcal R_{t,v}
 \leq\frac{t^2}{\eta_0}
       \log\E e^{\eta_0\mathcal R_{t,v}/t^2}
 \leq\frac{Dt+(\log C_1)t^2}{\eta_0}
 \leq C_2t,
\]
where $C_2:=(D+\log C_1)/\eta_0$.  This proves
\eqref{eq:bridge-remainder-estimates}.
\end{proof}
\begin{lemma}\label{lem:cir-joint}
The laws of $U_{t,v}$, for $v>0$, form a version of the conditional law
of $I_t^{\mathrm H}$ given $v_t=v$.  For $t>0$ and $v\geq0$, the laws of
$U_{t,v}$ have densities $f(t,v,i)$ that are continuous on $(0,\infty)\times[0,\infty)\times\R$.
They vanish for $i\leq0$ and are strictly positive for $i>0$.
Consequently, for every $t>0$, $(v_t,I_t^{\mathrm H})$ has joint density
$p(t,v)f(t,v,i)$ on $(0,\infty)^2$.
\end{lemma}

\begin{proof}
The conditional-law assertion follows from Glasserman and
Kim~\cite[Theorem~2.2]{GlassermanKim2011}, whose expansion holds for every
$\delta>0$, with their $\sigma=\xi$ and $\delta=2\alpha$.
Their Laplace-transform formulas~\cite[Lemma~2.4]{GlassermanKim2011},
combined with \eqref{eq:bessel-mixture}, give
\begin{equation}
 \begin{aligned}
 \mathcal L(b;t,v)&:=\E[e^{-bU_{t,v}}]
 =c(t,b)^\alpha
   \frac{\mathcal I_{\alpha-1}(w(t,v)c(t,b))}
        {\mathcal I_{\alpha-1}(w(t,v))}\\
 &\quad{}\times\exp\!\left\{
   \frac{v_0+v}{\xi^2}
   \left(\kappa\coth\frac{\kappa t}{2}
         -\omega(b)\coth\frac{\omega(b)t}{2}\right)\right\},
         \qquad (t>0,\ v,b\geq0),
 \end{aligned}
 \label{eq:bridge-laplace-transform}
\end{equation}
where, for $b\geq0$,
\[
 \omega(b):=\sqrt{\kappa^2+2\xi^2b},\qquad
 c(t,b):=\frac{\omega(b)\sinh(\kappa t/2)}
                  {\kappa\sinh(\omega(b)t/2)}.
\]
For every $b\geq0$, $\mathcal L(b;t,v)$ is continuous for $(t,v)\in(0,\infty)\times[0,\infty)$, since
$\mathcal I_{\alpha-1}$ is continuous and strictly positive on $[0,\infty)$.
The continuity theorem for Laplace transforms therefore shows that
$(t,v)\mapsto\operatorname{Law}(U_{t,v})$ is weakly continuous on
$(0,\infty)\times[0,\infty)$.

For the decomposition \eqref{eq:gamma-block}, independence and
\eqref{eq:bridge-laplace-transform} give
\[
 \E[e^{-b\mathcal R_{t,v}}]
 =\mathcal L(b;t,v)
   \prod_{n=1}^{M}\left(1+\frac{b}{\gamma_n(t)}\right)^\alpha,
 \qquad b\geq0.
\]
This is also continuous in $(t,v)$, so the same theorem shows that
$(t,v)\mapsto\operatorname{Law}(\mathcal R_{t,v})$ is weakly continuous.
Let $g^{(M)}(t,\cdot)$ denote the density of $Y_t$, extended by zero for $y\leq0$.
Next we establish continuity of $g^{(M)}(t,y)$ in $(t,y)\in(0,\infty)^2$.
For $M=1$, it is a scaled gamma density and the stated continuity follows.
For $M\geq2$, $g^{(M)}(t,\cdot)$ is the convolution of the densities of
$G_n/\gamma_n(t)$, $1\leq n\leq M$.  The substitution $y_n=yu_n$
in the convolution integral gives, for $y>0$,
\[
 g^{(M)}(t,y)
 =\frac{y^{r-1}\prod_{n=1}^{M}\gamma_n(t)^\alpha}{\Gamma(\alpha)^M}
   \int_{\Delta_M}
      e^{-y\sum_{n=1}^{M}\gamma_n(t)u_n}
      \prod_{n=1}^{M}u_n^{\alpha-1}\,
      \dd u_1\cdots\dd u_{M-1},
\]
where $u_M:=1-\sum_{n=1}^{M-1}u_n$ and
$\Delta_M:=\{(u_1,\ldots,u_{M-1})\in(0,\infty)^{M-1}:
\sum_{n=1}^{M-1}u_n<1\}$.
The integrand in this simplex integral is continuous in $(t,y)$ and
locally dominated by a constant multiple of the integrable function
$\prod_{n=1}^{M}u_n^{\alpha-1}$.  Dominated convergence therefore gives
continuity for $t,y>0$.

On the simplex, $\gamma_1(t)\leq\sum_{n=1}^{M}\gamma_n(t)u_n
\leq\gamma_M(t)$, while the Dirichlet integral is
\[
 \int_{\Delta_M}\prod_{n=1}^{M}u_n^{\alpha-1}
     \,\dd u_1\cdots\dd u_{M-1}
   =\frac{\Gamma(\alpha)^M}{\Gamma(M\alpha)}.
\]
Together, these observations yield the following bounds:
\begin{equation}
 \frac{\prod_{n=1}^{M}\gamma_n(t)^\alpha}{\Gamma(r)}
 y^{r-1}e^{-\gamma_M(t)y}
 \leq g^{(M)}(t,y)\leq
 \frac{\prod_{n=1}^{M}\gamma_n(t)^\alpha}{\Gamma(r)}
 y^{r-1}e^{-\gamma_1(t)y},\quad y>0.
 \label{eq:gamma-block-bounds}
\end{equation}
Since $r>1$, the upper bound in \eqref{eq:gamma-block-bounds}, together
with continuity for $t,y>0$, shows that $g^{(M)}(t,\cdot)$ is
bounded and continuous on $\R$ and vanishes at infinity.
For every $0<a<b<\infty$, the bounds in
\eqref{eq:gamma-block-bounds} imply that
$\sup_{t\in[a,b]}g^{(M)}(t,y)\to0$ as $y\downarrow0$ and as $y\to\infty$.
Together with continuity, this gives, for every $t_0>0$,
\begin{equation}
 \lim_{t\to t_0}\sup_{y\in\R}
 \left|g^{(M)}(t,y)-g^{(M)}(t_0,y)\right|=0.
 \label{eq:gamma-block-uniform-continuity}
\end{equation}
Since $Y_t$ and $\mathcal R_{t,v}$ are independent and $Y_t$ has density
$g^{(M)}(t,\cdot)$, their sum $U_{t,v}$ has density
\begin{equation}
 f(t,v,i):=\E[g^{(M)}(t,i-\mathcal R_{t,v})].
 \label{eq:bridge-convolution}
\end{equation}
Continuity of this density follows from weak continuity of
$(t,v)\mapsto\operatorname{Law}(\mathcal R_{t,v})$, the uniform convergence in
\eqref{eq:gamma-block-uniform-continuity}, and uniform continuity of each
zero-extended $g^{(M)}(t,\cdot)$.
It vanishes for $i\leq0$ because $\mathcal R_{t,v}\geq0$ and
$g^{(M)}(t,\cdot)$ vanishes on $(-\infty,0]$.
This also proves that $f$ is uniformly bounded for $t\in[a,b]$, for every
$0<a<b<\infty$.

By \eqref{eq:bessel-mixture}, $\mathbb P(\mathsf B=0)>0$.
Conditional on $\mathsf B=0$, write $\mathcal R_{t,v}$ as a convergent
sum of independent nonnegative blocks $(X_k)_{k\geq1}$, each having zero
in its support.  Write $\mathbb P_0:=\mathbb P(\,\cdot\mid\mathsf B=0)$.
Since the series converges $\mathbb P_0$-almost surely, its tails tend to
zero $\mathbb P_0$-almost surely.  Fix $i>0$ and choose $N\geq1$ such that
\[
 \mathbb P_0\!\left(\sum_{k>N}X_k<i/2\right)>\tfrac12.
\]
For each $k\leq N$, zero belongs to the support of $X_k$, so
$\mathbb P_0(X_k<i/(2N))>0$.  Thus, by independence of the $X_k$,
\[
 \mathbb P_0\!\left(\sum_{k\leq N}X_k<i/2\right)
 \geq\prod_{k=1}^N\mathbb P_0\!\left(X_k<i/(2N)\right)>0.
\]
Independence of the head and tail of the series therefore yields
\[
 \mathbb P_0(\mathcal R_{t,v}<i)
 \geq \mathbb P_0\!\left(\sum_{k\leq N}X_k<i/2\right)
       \mathbb P_0\!\left(\sum_{k>N}X_k<i/2\right)>0.
\]
Thus $\mathbb P(\mathcal R_{t,v}<i)>0$, and the lower bound in
\eqref{eq:gamma-block-bounds} and \eqref{eq:bridge-convolution} give
\[
 f(t,v,i)
 =\E\!\left[g^{(M)}(t,i-\mathcal R_{t,v});\,
             \mathcal R_{t,v}<i\right]>0.
\]
The identified conditional law and \eqref{eq:cir-density} give the
stated joint density.
\end{proof}

\subsection{The conditional variance and local variance surface}

The proposition below establishes the existence of a positive, continuous version of the conditional variance away from zero.

\begin{proposition}
\label{prop:gap-densities}
With $Q_t$, $p(t,\cdot)$ and $f(t,\cdot,\cdot)$ defined by \eqref{eq:gap-process},
\eqref{eq:cir-density} and \eqref{eq:bridge-convolution}, respectively,
define for $t,q>0$
\begin{equation}
 \begin{aligned}
 g(t,q)&:=\frac1\beta\int_0^q
     p(t,v)f\!\left(t,v,\frac{q-v}{\beta}\right)\dd v,\\
 m(t,q)&:=\frac1\beta\int_0^q
     v\,p(t,v)f\!\left(t,v,\frac{q-v}{\beta}\right)\dd v.
 \end{aligned}
 \label{eq:gap-densities}
\end{equation}
These are continuous, strictly positive densities of
$\Law(Q_t)$ and $\E[v_t;Q_t\in\dd q]$, respectively.
The ratio
\begin{equation}
 \phi(t,q):=\frac{m(t,q)}{g(t,q)},\qquad t,q>0,
 \label{eq:phi-definition}
\end{equation}
is a continuous version of $\E[v_t\mid Q_t=q]$, with
$0<\phi(t,q)\leq q$.
\end{proposition}

\begin{proof}
Apply the change of variables $(v,i)\mapsto(q,i)=(v+\beta i,i)$ to the
joint density in Lemma~\ref{lem:cir-joint}.  This gives the asserted densities and
the conditional identity.  Fix $0<t_-<t_+<\infty$ and
$0<q_-<q_+<\infty$.  For $t\in[t_-,t_+]$ and $q\in[q_-,q_+]$,
\eqref{eq:cir-boundary-domination} and the boundedness of $f$ imply that
the integrands in \eqref{eq:gap-densities} are bounded by a constant times
$v^{\alpha-1}$ on a fixed bounded
interval.  We extend each integrand by zero for $v\geq q$ and integrate
over this fixed interval.  As $(t_n,q_n)\to(t,q)$, continuity of $p$ and
of the zero-extended density $f$ implies that the extended integrands
converge for every $v>0$.  Dominated convergence
therefore proves continuity.  Positivity follows
since the integrands are strictly positive for $0<v<q$.
Finally, $0<m(t,q)\leq qg(t,q)$, which gives the bound at every $q>0$.
\end{proof}

The following lemma establishes the regularity of the conditional expectation $\phi(t,q)$ as $t\downarrow0$.

\begin{lemma}\label{lem:U0}
For $\phi$ defined in \eqref{eq:phi-definition} and every $0<q_-<q_+<\infty$,
\begin{equation}
 \lim_{t\downarrow0}\sup_{q\in[q_-,q_+]}
       |\phi(t,q)-q|=0.
 \label{eq:phi-small-time-convergence}
\end{equation}
\end{lemma}

\begin{proof}
By Lemma~\ref{lem:bridge-small-time} with $V=q_+$, the estimates
\eqref{eq:bridge-exp-moment} and \eqref{eq:bridge-remainder-estimates}
hold.
By \eqref{eq:gap-process} and Proposition~\ref{prop:gap-densities}, the
conditional expectation of $I_t^{\mathrm H}$ given $Q_t=q$ admits the
continuous version
\[
 e(t,q):=\frac{q-\phi(t,q)}{\beta},\qquad t,q>0.
\]
Since $q-\phi(t,q)=\beta e(t,q)\geq0$, it is enough to prove that
$e(t,q)\to0$ uniformly for $q\in[q_-,q_+]$.  The conditional density of
$I_t^{\mathrm H}$ given $Q_t=q$ is
\begin{equation}
 \frac{p(t,q-\beta i)f(t,q-\beta i,i)}{g(t,q)}
       \boldsymbol1_{\{0<i<q/\beta\}}.
 \label{eq:conditional-integral-density}
\end{equation}
We first derive an upper bound for its numerator
and then a lower bound for the denominator.  Throughout the remainder of
the proof, fix $0<t\leq1$, $q\in[q_-,q_+]$, and $0\leq v\leq q_+$.
All constants below depend only on
$\kappa,\theta,\xi,v_0,q_-,q_+$ and the fixed integer $M$ in
\eqref{eq:gamma-block}, but not on $t$, $q$, or $v$.
The upper bound in \eqref{eq:gamma-block-bounds}, with $y=t^2z$ and
$\gamma_1(t)t^2\geq c_-=2\eta_0$, yields
\begin{equation}
 \sup_{y\geq0}e^{\eta_0 y/t^2}g^{(M)}(t,y)\leq Ct^{-2}.
 \label{eq:weighted-gamma-block-bound}
\end{equation}
The identity \eqref{eq:bridge-convolution} expresses $f(t,v,\cdot)$
as the convolution of the density $g^{(M)}(t,\cdot)$ of $Y_t$ with the law of
the independent remainder $\mathcal R_{t,v}$.  Since $g^{(M)}(t,\cdot)$ is extended
by zero on $(-\infty,0]$, applying \eqref{eq:weighted-gamma-block-bound} with
$y=i-\mathcal R_{t,v}$ gives
\[
 \begin{aligned}
 e^{\eta_0 i/t^2}f(t,v,i)
 &=\E\!\left[e^{\eta_0\mathcal R_{t,v}/t^2}
       e^{\eta_0(i-\mathcal R_{t,v})/t^2}
       g^{(M)}(t,i-\mathcal R_{t,v})\right]\\
 &\leq Ct^{-2}\E e^{\eta_0\mathcal R_{t,v}/t^2}.
 \end{aligned}
\]
The first estimate in \eqref{eq:bridge-remainder-estimates} therefore yields
\begin{equation}
 e^{\eta_0 i/t^2}f(t,v,i)\leq Ct^{-2}e^{D/t},
 \qquad v\in[0,q_+],\quad i\geq0.
 \label{eq:weighted-bridge-density}
\end{equation}
Consequently, for $q\in[q_-,q_+]$, the numerator in
\eqref{eq:conditional-integral-density} satisfies
\begin{equation}
 \begin{aligned}
 &\int_0^{q/\beta}e^{\eta_0 i/t^2}
       p(t,q-\beta i)f(t,q-\beta i,i)\dd i\\
 &\qquad\leq\frac C\beta t^{-2}e^{D/t}\int_0^q p(t,v)\dd v
 \leq\frac C\beta t^{-2}e^{D/t}.
 \end{aligned}
 \label{eq:conditional-exp-numerator-bound}
\end{equation}
We next derive a lower bound for the denominator.  Set
$R_0:=\max\{1,4C_2\}$.  The second estimate in
\eqref{eq:bridge-remainder-estimates} and Markov's inequality give
$\mathbb P(\mathcal R_{t,v}\leq R_0t/2)\geq1/2$, so
for $R_0t\leq i\leq2R_0t$, nonnegativity of $\mathcal R_{t,v}$ gives
\begin{equation}
 \left\{\mathcal R_{t,v}\leq\frac{R_0t}{2}\right\}
 \subseteq
 \left\{\frac{R_0t}{2}\leq i-\mathcal R_{t,v}\leq2R_0t\right\}.
 \label{eq:remainder-event-inclusion}
\end{equation}
Restricting the expectation in \eqref{eq:bridge-convolution} to the event
on the left-hand side of \eqref{eq:remainder-event-inclusion}, which has
probability at least $1/2$, and applying the lower bound in
\eqref{eq:gamma-block-bounds} gives
\begin{equation}
 f(t,v,i)\geq c\,t^{-r-1}e^{-D_1/t},
 \qquad v\in[0,q_+],\quad R_0t\leq i\leq2R_0t.
 \label{eq:bridge-density-lower}
\end{equation}

Since $d(t)$ is bounded above and below by positive multiples of $t$,
\eqref{eq:cir-density} yields
\begin{equation}
 \inf_{v\in[q_-/2,q_+]}p(t,v)\geq c\,t^{-\alpha}e^{-D_2/t}.
 \label{eq:cir-density-lower}
\end{equation}
Fix $t_{\mathcal K}\leq1$ so that $2\beta R_0t_{\mathcal K}\leq q_-/2$.
In the remaining estimates, take $0<t\leq t_{\mathcal K}$ and
$q\in\mathcal K=[q_-,q_+]$.
Using the representation
$g(t,q)=\int_0^{q/\beta}p(t,q-\beta i)f(t,q-\beta i,i)\dd i$,
we obtain a lower bound by integrating
over $i\in[R_0t,2R_0t]$.  On this interval,
$q-\beta i\in[q_-/2,q_+]$, so
combining \eqref{eq:bridge-density-lower} and \eqref{eq:cir-density-lower} yields
\begin{equation}
 g(t,q)\geq c\,t^{-\alpha-r}e^{-D_3/t}
          \geq c e^{-D_3/t},\qquad q\in\mathcal K,\quad t\leq t_{\mathcal K}.
 \label{eq:gap-density-lower}
\end{equation}
Combining \eqref{eq:conditional-exp-numerator-bound} with
\eqref{eq:gap-density-lower} in
\eqref{eq:conditional-integral-density} gives constants
$A_{\mathcal K}\geq1$ and $B_{\mathcal K}>0$ such that
\[
 \sup_{q\in\mathcal K}\E[e^{\eta_0 I_t^{\mathrm H}/t^2}\mid Q_t=q]
       \leq A_{\mathcal K}t^{-2}e^{B_{\mathcal K}/t},
       \qquad 0<t\leq t_{\mathcal K}.
\]
Conditional Jensen's inequality now implies
\[
 e(t,q)\leq\frac{t^2}{\eta_0}
   \log\E[e^{\eta_0 I_t^{\mathrm H}/t^2}\mid Q_t=q]
 \leq\frac{B_{\mathcal K}}{\eta_0}t+
      \frac{t^2}{\eta_0}\bigl(\log A_{\mathcal K}+2\log(1/t)\bigr)
 \leq C_{\mathcal K}t.
\]
Since $q-\phi(t,q)=\beta e(t,q)$, the convergence in
\eqref{eq:phi-small-time-convergence} follows.
\end{proof}

\begin{proof}[Proof of Proposition~\ref{prop:surface}]
Extend $\phi(t,\cdot)$, defined in \eqref{eq:phi-definition}, by zero for $q\leq0$,
and set $\phi(0,q):=q^+$.
Using $A$ from \eqref{eq:beta-A}, define
\begin{equation}
 \lambda(t,s):=\phi(t,A(t,s)),\qquad
 \mu(t,x):=\lambda(t,e^x),\qquad t\geq0,\ s>0.
 \label{eq:canonical-coefficients}
\end{equation}
For \textup{(i)}, Proposition~\ref{prop:gap-densities} gives a version of
$\E[v_t\mid Q_t=q]$.  Since $Q_t=A(t,S_t)$ and, for fixed $t$, the map
$s\mapsto A(t,s)$ is one-to-one, \eqref{eq:canonical-coefficients} gives a
version of $\E[v_t\mid S_t=s]$.

For \textup{(ii)}, Proposition~\ref{prop:gap-densities} gives continuity of
$\phi$ for $t,q>0$, while Lemma~\ref{lem:U0} gives continuity at $t=0$ for
$q>0$.  At $q=0$, the bound $0\leq\phi(t,q)\leq q^+$ gives continuity of
the zero extension.  Thus $\phi$ is continuous on $[0,\infty)\times\R$.
Composition with $A$ proves the asserted continuity, and
\[
 \lambda(0,s)=\phi(0,A(0,s))=A(0,s)^+
              =\bigl(v_0-\xi\log(s/S_0)\bigr)^+.
\]

For \textup{(iii)}, $A(t,s)>0$ exactly when $0<s<s^*(t)$.
Proposition~\ref{prop:gap-densities} gives $\phi(t,q)>0$ for $t,q>0$,
the definition $\phi(0,q)=q^+$ covers $t=0$, and the zero extension covers
$A(t,s)\leq0$.  Hence $\lambda(t,s)>0$ below the barrier and vanishes on
and above it.

Finally, for \textup{(iv)}, the bound $0\leq\phi(t,q)\leq q^+$ gives
$0\leq\lambda(t,s)\leq A(t,s)^+$, and, for $t\in[0,T]$,
\[
 A(t,e^x)^+\leq
 v_0+\kappa\theta T+\xi|\log S_0|+\xi|x|.
\]
This yields the stated linear-growth bound.
\end{proof}

%% file: final_appendix_uniqueness.tex
\section{Uniqueness in law}
\label{app:uniqueness}
In this appendix we prove uniqueness in law for the projected SDE.

\begin{lemma}
Let $Q_t$ be defined by \eqref{eq:gap-process}.  For every finite $T>0$,
\begin{equation}
 \sup_{0\leq t\leq T}\E[Q_t^{-1}]<\infty.
 \label{eq:gap-negative-moments}
\end{equation}
\end{lemma}

\begin{proof}
Put $M_t:=\int_0^t\sqrt{v_u}\,\dd Z_u$, so that
$\langle M\rangle_t=I_t^{\mathrm H}$.  Since $\beta-\kappa=\xi/2$,
\[
 Q_t=v_0+\kappa\theta t+\tfrac\xi2 I_t^{\mathrm H}+\xi M_t,
 \qquad Q_t\geq\beta I_t^{\mathrm H}.
\]
Fix $t\geq0$ and $0<q<v_0$.  On $\{Q_t\leq q\}$, we therefore have
$M_t\leq-(v_0-q)/\xi$ and $I_t^{\mathrm H}\leq q/\beta$.
For $\ell>0$, the nonnegative exponential local martingale
$\exp\{-\ell M_t-\ell^2I_t^{\mathrm H}/2\}$ is a supermartingale
starting at one.  Hence
\[
 \mathbb P(Q_t\leq q)\leq
 \exp\!\left\{-\ell\frac{v_0-q}{\xi}+\frac{\ell^2q}{2\beta}\right\}.
\]
The right-hand side is minimized over $\ell>0$ at
$\ell=\beta(v_0-q)/(\xi q)$.  Substituting this value yields
\[
 \mathbb P(Q_t\leq q)\leq
 \exp\!\left\{-\frac{\beta(v_0-q)^2}{2\xi^2q}\right\}.
\]
For $q\leq v_0/2$, the right-hand side is
bounded by $e^{-d_0/q}$, where $d_0:=\beta v_0^2/(8\xi^2)>0$.
Thus
\[
 \E[Q_t^{-1}]
 =\int_0^\infty q^{-2}\mathbb P(Q_t<q)\dd q
 \leq\int_0^{v_0/2}q^{-2}e^{-d_0/q}\dd q+\frac2{v_0},
\]
uniformly in $t$.
\end{proof}

Write
\[
 x^*(t):=\log S_0+(v_0+\kappa\theta t)/\xi.
\]
For a strictly positive solution $\bar S$ of \eqref{eq:lv-S}, define
its log-price $\bar X_t:=\log\bar S_t$ and its gap process by
\[
 \bar Q_t:=\xi\bigl(x^*(t)-\bar X_t\bigr)
          =v_0+\kappa\theta t-\xi\log(\bar S_t/S_0),
 \qquad \bar Q_0=v_0.
\]
The gap is the scaled distance of the log-price below the moving
barrier $x^*(t)=\log s^*(t)$.
The log-price and gap forms of the projected equation are
\begin{align}
 \dd\bar X_t&=-\tfrac12\mu(t,\bar X_t)\,\dd t
              +\sqrt{\mu(t,\bar X_t)}\,\dd B_t,
       &\bar X_0&=\log S_0,\label{eq:lv-log}\\
 \dd\bar Q_t&=\left(\kappa\theta+\tfrac\xi2\phi(t,\bar Q_t)\right)\dd t
             +\xi\sqrt{\phi(t,\bar Q_t)}\,\dd\widetilde B_t,
       &\bar Q_0&=v_0,\label{eq:projected-gap-sde}
\end{align}
where $\widetilde B=-B$.
The zero extension of $\phi$ makes the gap equation well defined on the
whole real line.
It\^o's formula gives
\[
 \dd\bar X_t
 =\frac{\dd\bar S_t}{\bar S_t}
   -\frac12\lambda(t,\bar S_t)\,\dd t
 =-\frac12\mu(t,\bar X_t)\,\dd t
   +\sqrt{\mu(t,\bar X_t)}\,\dd B_t,
\]
which is \eqref{eq:lv-log}.  Conversely, if $\bar X$ solves
\eqref{eq:lv-log}, then $\bar S_t=e^{\bar X_t}>0$ and
\[
 \dd\bar S_t
 =\bar S_t\,\dd\bar X_t
   +\frac12\bar S_t\mu(t,\bar X_t)\,\dd t
 =\bar S_t\sqrt{\lambda(t,\bar S_t)}\,\dd B_t,
\]
so $\bar S$ solves \eqref{eq:lv-S}.
Moreover, $\mu(t,\bar X_t)=\phi(t,\bar Q_t)$ and
$\xi(x^*)'(t)=\kappa\theta$.  Thus the affine change of variables
$\bar Q_t=\xi(x^*(t)-\bar X_t)$ gives
\[
 \begin{aligned}
 \dd\bar Q_t
 &=\kappa\theta\,\dd t-\xi\,\dd\bar X_t\\
 &=\left(\kappa\theta+\frac\xi2\phi(t,\bar Q_t)\right)\dd t
   -\xi\sqrt{\phi(t,\bar Q_t)}\,\dd B_t,
 \end{aligned}
\]
which is \eqref{eq:projected-gap-sde} with $\widetilde B=-B$.
The substitution $\bar X_t=x^*(t)-\bar Q_t/\xi$, with
$B=-\widetilde B$, recovers \eqref{eq:lv-log}.
These transformations give a one-to-one correspondence between continuous real-valued
solutions of the log-price and gap equations and continuous strictly
positive solutions of the spot equation.

\begin{proposition}
\label{prop:barrier-nonattainment}
Let $\bar Q$ be a continuous mimicking solution of
\eqref{eq:projected-gap-sde} on $[0,T]$, that is, a continuous weak
solution with $\bar Q_0=v_0$ and $\Law(\bar Q_t)=\Law(Q_t)$ for every
$t\in[0,T]$.  Then $\mathbb P(\bar Q_t>0\text{ for all }t\in[0,T])=1$;
equivalently, the corresponding solution $\bar S$ of \eqref{eq:lv-S}
satisfies $\bar S_t<s^*(t)$ for all $t\in[0,T]$, almost surely.
\end{proposition}

\begin{proof}
Mimicking, Tonelli's theorem and \eqref{eq:gap-negative-moments} give
\begin{equation}
 \E\int_0^T\frac{\dd t}{\bar Q_t^+}
   =\int_0^T\E[Q_t^{-1}]\,\dd t<\infty,
 \label{eq:projected-gap-inverse-integral}
\end{equation}
where $1/0=+\infty$.
Set $\tau:=\inf\{t\in[0,T]:\bar Q_t=0\}$, $\inf\emptyset=\infty$, and
$\phi_u:=\phi(u,\bar Q_u)$.  Since $\phi(u,q)\leq q$ by
Proposition~\ref{prop:gap-densities}, the process
\[
 \Lambda_t:=\log v_0
 +\int_0^{t}\boldsymbol1_{\{u<\tau\}}\left[
       \frac{\kappa\theta+\frac\xi2\phi_u}{\bar Q_u}
       -\frac{\xi^2\phi_u}{2\bar Q_u^2}\right]\dd u
 +\xi\int_0^{t}\boldsymbol1_{\{u<\tau\}}
       \frac{\sqrt{\phi_u}}{\bar Q_u}\,\dd\widetilde B_u,
 \qquad t\in[0,T],
\]
has drift integrand bounded in absolute value by
$(\kappa\theta+\xi^2/2)/\bar Q_u^{+}+\xi/2$ and martingale integrand
bounded by $\xi/\sqrt{\bar Q_u^{+}}$; by
\eqref{eq:projected-gap-inverse-integral} the drift integral converges
absolutely and the stochastic integral is a square-integrable
martingale, so $\Lambda$ is almost surely finite and continuous on
$[0,T]$.  It\^o's formula for $\log$ on the stochastic interval
$[0,\tau)$, where $\bar Q>0$ (localize at the hitting times of $1/n$),
gives $\log\bar Q_t=\Lambda_t$ for all $t\in[0,\tau)\cap[0,T]$, almost
surely.  On $\{\tau\leq T\}$, letting $t\uparrow\tau$ gives
$\log\bar Q_t\to\log\bar Q_\tau=-\infty$, whereas
$\Lambda_t\to\Lambda_\tau\in\R$; hence $\mathbb P(\tau\leq T)=0$, that
is, $\bar Q>0$ on $[0,T]$ almost surely.  Equivalently, $\bar S<s^*$ on
$[0,T]$, because $\bar Q_t=\xi\log\bigl(s^*(t)/\bar S_t\bigr)$.
\end{proof}

\begin{proof}[Proof of Theorem~\ref{thm:timezero}]
Fix a weak mimicking solution $\bar S^\star$ of \eqref{eq:lv-S}
supplied by Lemma~\ref{lem:mimicking-existence}, and set
$\bar X_t^\star:=\log\bar S_t^\star$.  The corresponding gap process is 
\[
 \bar Q_t^\star:=\xi\bigl(x^*(t)-\bar X_t^\star\bigr)
   =v_0+\kappa\theta t-\xi\log(\bar S_t^\star/S_0),
 \qquad \bar Q_0^\star=v_0.
\]
The transformation above shows that $\bar Q^\star$ is a weak mimicking
solution of \eqref{eq:projected-gap-sde}.   By
Proposition~\ref{prop:barrier-nonattainment}, its paths are continuous
and strictly positive on $[0,T]$.

Work on the complete separable state space $E:=\R\times[0,\infty)$.
For a sufficiently large integer $n$ with $v_0\in(1/n,n)$, set
\[
 \mathfrak c_n(q):=\max\{1/n,\min\{q,n\}\},\qquad
 \phi^{(n)}(q,t):=\phi(t\wedge T,\mathfrak c_n(q)),\qquad (q,t)\in E.
\]
Here $\mathfrak c_n$ replaces state arguments outside $[1/n,n]$ by the
nearest endpoint, while the time argument is set to $T$ for $t>T$.
Proposition~\ref{prop:surface} and
Lemma~\ref{lem:U0} give
\[
 \underline a_n:=\xi^2
 \min_{(y,u)\in[1/n,n]\times[0,T]}\phi(u,y)>0.
\]
Thus the drift $b^{(n)}:=\kappa\theta+\xi\phi^{(n)}/2$ and
covariance $a^{(n)}:=\xi^2\phi^{(n)}$ are bounded and continuous
on $E$, and for each fixed $n$,
\[
 a^{(n)}(q,t)\geq\underline a_n>0,
 \qquad (q,t)\in E.
\]
Define the time-dependent operator
\[
 \mathcal L_t^{(n)}f(q):=b^{(n)}(q,t)f'(q)
                  +\tfrac12a^{(n)}(q,t)f''(q),
 \qquad f\in C_c^\infty(\R).
\]
The boundedness, continuity and uniform ellipticity above imply, by
\cite[Theorem~8.1.7, pp.~370--371]{EthierKurtz1986}, that the
martingale problem for $\mathcal L_t^{(n)}$ is well posed.
Define the corresponding space--time operator by
\begin{equation}
 \mathcal A^{(n)}:=\left\{
   \bigl(f\gamma,\,\gamma\mathcal L_t^{(n)}f+f\gamma'\bigr):
   f\in C_c^\infty(\R),\ \gamma\in C_c^1([0,\infty))
 \right\}.
 \label{eq:cutoff-generator}
\end{equation}
For any solution $(Y_u,R_u)$ with initial law $\delta_{(v_0,0)}$, the
time-coordinate martingale identities force $R_u=u$.
By \cite[Theorem~4.7.1, p.~221]{EthierKurtz1986}, $Y$ solves the
martingale problem for $\mathcal L_t^{(n)}$ with $Y_0=v_0$.
Uniqueness for the martingale problem for
$\mathcal L_t^{(n)}$ therefore gives uniqueness for
$(\mathcal A^{(n)},\delta_{(v_0,0)})$.
Well-posedness of the time-dependent martingale problem for $\mathcal L_t^{(n)}$
gives a solution of
the martingale problem for $\mathcal A^{(n)}$ for every initial law on $E$.
Let $\bar Q$ be any continuous real-valued weak solution of
\eqref{eq:projected-gap-sde} from $v_0$.  Set
\[
 U_n:=(1/n,n)\times[0,T),
 \qquad
 \eta_n:=\inf\{t\in[0,T]:\bar Q_t\notin(1/n,n)\},
\]
with $\inf\varnothing:=T$.  The sets $U_n$ are increasing and open
relative to $E$.  Define $\eta_n^\star$ analogously for $\bar Q^\star$.
By It\^o's formula,
the processes
$(\bar Q_{u\wedge\eta_n},u\wedge\eta_n)_{u\geq0}$ and
$(\bar Q_{u\wedge\eta_n^\star}^\star,u\wedge\eta_n^\star)_{u\geq0}$
both solve the stopped martingale problem for
$(\mathcal A^{(n)},\delta_{(v_0,0)},U_n)$.
By \cite[Lemma~4.5.16, p.~206]{EthierKurtz1986}, each of the two stopped processes
can be represented in law by stopping a solution $Z$ of
$(\mathcal A^{(n)},\delta_{(v_0,0)})$ at some random time $\tau$.
The stopped process lies in $U_n$ before its stopping time and outside
$U_n$ at that time.  The same therefore holds for $(Z,\tau)$, so $\tau$
is the exit time of $Z$ from $U_n$, a measurable functional of $Z$.
Uniqueness for $(\mathcal A^{(n)},\delta_{(v_0,0)})$ therefore gives
\[
 \Law\bigl((\bar Q_{u\wedge\eta_n},u\wedge\eta_n)_{u\geq0},\eta_n\bigr)
 =\Law\bigl((\bar Q_{u\wedge\eta_n^\star}^\star,u\wedge\eta_n^\star)_{u\geq0},
            \eta_n^\star\bigr).
\]

Since $\bar Q^\star$ has continuous paths, by Proposition~\ref{prop:barrier-nonattainment}
$\bar Q^\star$ attains a positive minimum and a finite maximum
on $[0,T]$, almost surely.  Hence $\eta_n^\star=T$ for all sufficiently
large $n$, almost surely.  Writing $\mathbb P^\star$ and $\E^\star$ for
the reference probability and corresponding expectation, equality of
the stopped laws gives
\[
 \mathbb P(\eta_n<T)=\mathbb P^\star(\eta_n^\star<T)
       \longrightarrow0.
\]
The full and stopped paths agree when their stopping time equals $T$.
Thus, for every bounded Borel functional $F$ on $C([0,T];\R)$,
equality of the stopped laws and the triangle inequality give
\[
 \left|\E F(\bar Q)-\E^\star F(\bar Q^\star)\right|
 \leq 2\|F\|_\infty
       \bigl(\mathbb P(\eta_n<T)+\mathbb P^\star(\eta_n^\star<T)\bigr)
 \longrightarrow0.
\]
Hence $\bar Q$ and $\bar Q^\star$ have the same law on $C([0,T];\R)$.
Transforming back to spot proves the first claims.
\end{proof}


%% file: final_appendix_comparison.tex
\section{Smooth comparison and approximation}
\label{app:comparison}

This appendix proves the smooth comparison and approximation results used
in the proofs of Theorems~\ref{thm:main} and~\ref{thm:strict}.

\subsection{The smooth comparison}

For each starting time $r\in[0,T]$, consider the extended state
\begin{equation}
\begin{aligned}
 \dd X_u&=-\tfrac12\nu(u,X_u)\,\dd u+\sqrt{\nu(u,X_u)}\,\dd B_u,\\
 \dd I_u&=\nu(u,X_u)\,\dd u,
 \qquad (X_r,I_r)=(x,i),
\end{aligned}
 \label{eq:extended-state}
\end{equation}
for a mapping $\nu$ on $[0,T]\times\R$, and let
\[
 \Dop:=\tfrac12(\partial_{xx}-\partial_x)+\partial_i,
\]
so that the backward generator of \eqref{eq:extended-state} is
$\nu(r,x)\Dop$.

\begin{lemma}\label{lem:smooth}
Let $\nu:[0,T]\times\R\to(0,\infty)$ be continuous and
$C^\infty$ in $x$, with $\inf\nu>0$ and, for some $C_\nu<\infty$,
\[
 \nu(r,x)\leq C_\nu(1+|x|).
\]
Assume that every positive-order spatial derivative of $\nu$ is bounded
and continuous.  Let $\psi\in C^\infty(\R)$ be nondecreasing and
of at most linear growth, with bounded derivatives of all positive orders,
and put
\begin{equation}
 U(r,x,i):=\E_{r,x}\Bigl[\psi\Bigl(i+\int_r^T\nu(u,X_u)\,\dd u\Bigr)\Bigr],
 \qquad h:=\Dop U.
 \label{eq:U-def}
\end{equation}
Then $U\in C^{1,2,1}$, with one-sided time derivatives at the endpoints,
and
\[
 U_r+\nu\Dop U=0,
 \qquad U(T,x,i)=\psi(i).
\]
The functions $U_x$ and $h$ are bounded, $h\geq0$ on
$[0,T]\times\R\times\R$, and, if $\psi$ is convex, the map
$i\mapsto h(r,x,i)$ is nondecreasing.  Moreover, for
$0\leq t < t'\leq T$,
\begin{equation}
 U(t,x,i)-U(t',x,i)=\int_t^{t'}\nu(r,x)\,h(r,x,i)\,\dd r.
 \label{eq:U-monotone}
\end{equation}
\end{lemma}

\begin{proof}
Put $\underline\nu:=\inf\nu>0$ and
$M_1:=\|\nu_x\|_\infty<\infty$.  Since
\[
 \bigl\|\partial_x(-\nu/2)\bigr\|_\infty\leq\frac{M_1}{2},
 \qquad
 \bigl\|\partial_x\sqrt\nu\bigr\|_\infty
 =\left\|\frac{\nu_x}{2\sqrt\nu}\right\|_\infty
 \leq\frac{M_1}{2\sqrt{\underline\nu}},
\]
the coefficients $-\nu/2$ and $\sqrt\nu$ are globally Lipschitz in $x$,
uniformly in $r\in[0,T]$.  They have linear growth, and their
spatial derivatives of all positive orders are bounded.  The smooth-flow
theorem and derivative estimates in
Kunita~\cite[Ch.~4, Thms.~4.6.4--4.6.5, pp.~172--175]{Kunita1990},
together with the Burkholder--Davis--Gundy inequality and Gronwall's lemma
applied inductively to the derivative equations, give the time-supremum bounds
\[
 \sup_{r\in[0,T]}\sup_{x\in\R}
 \E\Bigl[\sup_{u\in[r,T]}|\partial_x^jX_u^{r,x}|^p\Bigr]<\infty,
 \qquad j\geq1,\quad 1\leq p<\infty.
\]
The same estimates show that derivatives of all orders of
$\int_r^T\nu(u,X_u^{r,x})\dd u$ with respect to $x$
have finite moments of every order, uniformly over $r\in[0,T]$ and $x\in\R$.
Consequently, differentiation
under the expectation in \eqref{eq:U-def} is legitimate: $U$ is smooth
in $(x,i)$, and all its spatial derivatives of all orders are bounded
and continuous, including at $r=T$.  The function $U$ itself has at
most linear growth.

For $r<q\leq T$, the Markov property and It\^o's formula applied to the
fixed function $U(q,\cdot,\cdot)$ give
\begin{equation*}
 U(r,x,i)-U(q,x,i)
 =\E_{r,x,i}\int_r^q\nu(u,X_u)\Dop U(q,X_u,I_u)\,\dd u.
\end{equation*}
Localization is removed using the preceding derivative and moment
bounds.  Dividing by $q-r$ and taking the right limit, and then using
the analogous left difference, proves
\begin{equation}
 U_r=-\nu\Dop U=-\nu h,\qquad U(T,x,i)=\psi(i),
 \label{eq:U-pde}
\end{equation}
with $U_r$ continuous on $[0,T]\times\R^2$.

Spatial differentiation of the integral form of \eqref{eq:U-pde} is
justified by the locally bounded continuous derivatives just obtained.
Thus $h_r=-\Dop(\nu h)$, and we have
\begin{equation*}
 \Dop(\nu h)=\nu\Dop h+\nu_xh_x
                      +\tfrac12(\nu_{xx}-\nu_x)h.
\end{equation*}
It follows that $h\in C^{1,2,1}$ and
\begin{equation}
 \begin{aligned}
 h_r+\tfrac12\nu h_{xx}
   +(\nu_x-\tfrac12\nu)h_x+\nu h_i+\varpi h&=0,\\
 h(T,x,i)=\psi'(i),\qquad
 \varpi:=\tfrac12(\nu_{xx}-\nu_x)&.
 \end{aligned}
 \label{eq:h-pde}
\end{equation}
Since $\nu_x$ and $\nu_{xx}$ are bounded by assumption, $\varpi$ is bounded.

Consider the auxiliary diffusion
\begin{equation}
 \begin{aligned}
 \dd\widehat X_u
  &=(\nu_x-\tfrac12\nu)(u,\widehat X_u)\,\dd u
             +\sqrt{\nu(u,\widehat X_u)}\,\dd\widehat B_u,\\
 \dd\widehat I_u&=\nu(u,\widehat X_u)\,\dd u,
       \qquad(\widehat X_r,\widehat I_r)=(x,i).
 \end{aligned}
 \label{eq:aux}
\end{equation}
Its coefficients are globally Lipschitz in the state variables, uniformly
in time, and have linear growth.  Set
\[
 \mathcal E_u:=\exp\!\left\{\int_r^u\varpi(w,\widehat X_w)\,\dd w\right\}.
\]
It\^o's formula and \eqref{eq:h-pde} give
\[
 \mathcal E_u h(u,\widehat X_u,\widehat I_u)
 =h(r,x,i)+\int_r^u\mathcal E_w
 h_x(w,\widehat X_w,\widehat I_w)
 \sqrt{\nu(w,\widehat X_w)}\,\dd\widehat B_w.
\]
The stochastic integral on the right is a square-integrable martingale: $h_x$ and
$\varpi$ are bounded, and
$\widehat\E_{r,x}\int_r^T\nu(u,\widehat X_u)\dd u<\infty$.
Taking expectations and using continuity at $T$ gives
\begin{equation}
 h(r,x,i)=\widehat\E_{r,x}\!\left[
 e^{\int_r^T\varpi(u,\widehat X_u)\dd u}
 \psi'\!\left(i+\int_r^T\nu(u,\widehat X_u)\,\dd u\right)\right].
 \label{eq:h-FK}
\end{equation}
In particular,
$0\leq h\leq e^{T\|\varpi\|_\infty}\|\psi'\|_\infty$.
If $\psi$ is convex, then $\psi'$ is nondecreasing.
The law of $\widehat X$ in \eqref{eq:aux} is independent of $i$, so
\eqref{eq:h-FK} shows that $h$ is nondecreasing in $i$.
Finally, integration of \eqref{eq:U-pde} at fixed $(x,i)$
proves \eqref{eq:U-monotone}.
\end{proof}

\subsection{Removal of smoothness}

\begin{lemma}\label{lem:stability}
Let $\nu:[0,T]\times\R\to[0,\infty)$ be continuous and satisfy
\begin{equation}
 \nu(u,y)\leq C_\nu(1+|y|),\qquad (u,y)\in[0,T]\times\R,
 \label{eq:G}
\end{equation}
for some $C_\nu<\infty$.  Fix $(x,i)\in\R\times\R_+$.
Suppose that
\begin{equation}
 \dd X_u=-\tfrac12\nu(u,X_u)\,\dd u
          +\sqrt{\nu(u,X_u)}\,\dd B_u,
 \qquad X_0=x,
 \label{eq:limit-sde}
\end{equation}
has a weak solution unique in law on $[0,T]$, and put
$I_u=i+\int_0^u\nu(w,X_w)\,\dd w$.
There exist continuous functions $\nu_n$, smooth in the spatial
variable with bounded, continuous derivatives of every positive
order, and a constant $C_\nu'<\infty$ such that
\begin{equation}
 \begin{gathered}
  \nu_n\geq\nu+\tfrac1n,\qquad
  \nu_n(u,y)\leq C_\nu'(1+|y|)\quad\text{for all }n,\\
  \sup_{(u,y)\in[0,T]\times[-R,R]}
       |\nu_n(u,y)-\nu(u,y)|\longrightarrow0
       \quad\text{for every }R>0.
 \end{gathered}
 \label{eq:nu-n}
\end{equation}
If $(X^n,I^n)$ solves \eqref{eq:extended-state} with coefficient $\nu_n$
and initial state $(x,i)$, then
\begin{equation}
 (X^n,I^n)\Longrightarrow(X,I)
 \qquad\text{in }C([0,T];\R\times\R_+).
 \label{eq:conv}
\end{equation}
Moreover, for every continuous $F:\R\times\R_+\to\R$ of at most linear growth,
\[
 \E[F(X_T^n,I_T^n)]\longrightarrow\E[F(X_T,I_T)].
\]
\end{lemma}

\begin{proof}
Let $\chi\in C_c^\infty(\R)$ be nonnegative, supported in $[-1,1]$,
with $\int\chi=1$.  Set
\[
 c_\chi=\int_\R|z|\chi(z)\,\dd z,\qquad
 L_n=C_\nu+n,\qquad \delta_n=L_n^{-3},
 \qquad
 g_n(u,y)=\sup_{z\in\R}\{\nu(u,z)-L_n|y-z|\}.
\]
Then $g_n(u,\cdot)$ is $L_n$-Lipschitz and
\begin{equation}
 \nu(u,y)\leq g_n(u,y)\leq C_\nu(1+|y|).
 \label{eq:g-n-bounds}
\end{equation}
For $|y|\leq R$, the supremum is attained at a point satisfying
$|z-y|\leq C_\nu(1+R)/n$.  Hence $g_n$ is continuous and
converges to $\nu$ uniformly on $[0,T]\times[-R,R]$ for every $R>0$.
Define
\begin{equation}
 \nu_n(u,y)=\int_\R g_n(u,y-\delta_nz)\chi(z)\,\dd z
              +L_n\delta_n c_\chi+\frac1n.
 \label{eq:majorant-construction}
\end{equation}
The Lipschitz bound gives
\[
 \left|\int_\R g_n(u,y-\delta_nz)\chi(z)\,\dd z-g_n(u,y)\right|
       \leq L_n\delta_n c_\chi.
\]
Thus \eqref{eq:nu-n} follows from \eqref{eq:g-n-bounds} and
$L_n\delta_n\to0$.  Differentiation of the mollifier, using the
Lipschitz bound for $g_n$, gives the asserted derivative bounds and
continuity.

For each fixed $n$, the maps $y\mapsto-\nu_n(u,y)/2$ and
$y\mapsto\sqrt{\nu_n(u,y)}$ are globally Lipschitz on $\R$, uniformly
in $u\in[0,T]$.  Indeed, as
$\nu_n\geq1/n$, we have
\[
 \left|\partial_y\sqrt{\nu_n(u,y)}\right|
 =\frac{|\partial_y\nu_n(u,y)|}{2\sqrt{\nu_n(u,y)}}
 \leq\frac{\sqrt n}{2}\|\partial_y\nu_n\|_\infty.
\]
Both maps have linear growth by \eqref{eq:nu-n}.
The approximating equations therefore have unique
strong solutions.  The common growth bound, the
Burkholder--Davis--Gundy inequality and Gronwall's lemma, applied after
stopping at bounded levels, give
\begin{align}
 \sup_n\E\sup_{u\leq T}|X_u^n|^4
     &\leq C_T(1+|x|^4),
     \label{eq:approx-fourth-moment}\\
 \sup_n\E\sup_{u\leq T}|I_u^n|^4
     &\leq C_T(1+|x|^4+i^4),
     \label{eq:approx-integral-fourth-moment}
\end{align}
and, for $0\leq u\leq w\leq T$,
\begin{equation}
 \begin{aligned}
  \E|X_w^n-X_u^n|^4&\leq C_T(1+|x|^4)(w-u)^2,\\
  \E|I_w^n-I_u^n|^4&\leq C_T(1+|x|^4)(w-u)^4,
 \end{aligned}
 \label{eq:approx-increments}
\end{equation}
so that Kolmogorov's
criterion, Billingsley~\cite[Theorem~12.3, pp.~95--96]{Billingsley1999}, gives tightness of
$\{\Law(X^n,I^n)\}$ in $C([0,T];\R\times\R_+)$.

For $f\in C_c^\infty(\R^2)$, write
\[
 \begin{aligned}
 \mathcal A_n(u)f(y,j)
    &=\nu_n(u,y)\left(\tfrac12f_{yy}-\tfrac12f_y+f_j\right)(y,j),\\
 \mathcal A(u)f(y,j)
    &=\nu(u,y)\left(\tfrac12f_{yy}-\tfrac12f_y+f_j\right)(y,j).
 \end{aligned}
\]
Since $f$ has compact support, the local uniform convergence in \eqref{eq:nu-n}
implies
\[
 \sup_{(u,y,j)\in[0,T]\times\R^2}
       |\mathcal A_n(u)f(y,j)-\mathcal A(u)f(y,j)|\longrightarrow0.
\]
Ethier--Kurtz~\cite[Theorem~4.8.10, pp.~234--235]{EthierKurtz1986} now
shows that every weak limit
$(Y,J)$ solves the martingale problem for $\mathcal A$ with initial
state $(x,i)$.
The martingale-problem representation
\cite[Theorem~5.3.3, pp.~293--294]{EthierKurtz1986} implies that $Y$
solves \eqref{eq:limit-sde} and
\[
 J_u=i+\int_0^u\nu(r,Y_r)\,\dd r.
\]
Uniqueness in law for \eqref{eq:limit-sde} therefore gives
$\Law(Y,J)=\Law(X,I)$ and proves \eqref{eq:conv}.

By the Portmanteau theorem
(e.g.\ \cite[Theorem~3.3.1, pp.~108--110]{EthierKurtz1986}),
\eqref{eq:approx-fourth-moment} and
\eqref{eq:approx-integral-fourth-moment} also hold for $(X,I)$.
For every continuous $F:\R\times\R_+\to\R$ of at most linear growth,
uniform integrability and \eqref{eq:conv} give
\[
 \E[F(X_T^n,I_T^n)]\longrightarrow\E[F(X_T,I_T)].
\]
\end{proof}